\documentclass[aps,amsmath,amssymb,superscriptaddress,a4paper,onecolumn,accepted=2021-07-29]{quantumarticle}

\pdfoutput=1

\usepackage[utf8]{inputenc}
\usepackage[english]{babel}
\usepackage[T1]{fontenc}
\usepackage{amsmath}
\makeatletter
\newcommand{\pushright}[1]{\ifmeasuring@#1\else\omit\hfill$\displaystyle#1$\fi\ignorespaces}
\newcommand{\pushleft}[1]{\ifmeasuring@#1\else\omit$\displaystyle#1$\hfill\fi\ignorespaces}
\makeatother

\usepackage{tikz}
\usepackage[numbers,sort&compress,square]{natbib}
\usepackage{amssymb}
\usepackage{color}
\usepackage{mathrsfs}
\usepackage{amsbsy}
\usepackage{amsthm}
\usepackage{graphicx}
\usepackage{tikz}

\usepackage{bbm}
\usepackage{bm}
\usepackage{epsfig}
\usepackage{xfrac}
\usepackage{xcolor}
\usepackage{enumerate}
\usepackage[shortlabels]{enumitem}

\definecolor{myurlcolor}{rgb}{0,0,1}
\definecolor{myrefcolor}{rgb}{0,0,1}
\usepackage[unicode=true,pdfusetitle, bookmarks=true,bookmarksnumbered=false,bookmarksopen=false, breaklinks=false,pdfborder={0 0 0},backref=false,colorlinks=true, linkcolor=myrefcolor,citecolor=myurlcolor,urlcolor=myurlcolor]{hyperref}

\usepackage{braket}
\renewcommand{\v}[1]{\ensuremath{\mathbf{#1}}} 
\newcommand{\gv}[1]{\ensuremath{\text{\boldmath$ #1 $}}}
\newcommand{\abs}[1]{\left| #1 \right|} 
\newcommand{\norm}[1]{\left\| #1 \right\|} 
\newcommand{\ptheta}{{\partial_\theta}}
\newcommand{\trace}{\mathrm{Tr}}
\newcommand{\Choi}{\mathrm{Choi}}
\newcommand{\avg}{\mathrm{avg}}
\newcommand{\txs}{{\mathscr{S}}}
\newcommand{\txr}{{\mathscr{R}}}
\newcommand{\reg}{{{\rm reg}}}

\newcommand{\teps}{{\tilde\epsilon}}

\newcommand{\tPsi}{{\tilde\Psi}}

\newcommand{\mI}{{\mathcal{I}}}
\newcommand{\mA}{{\mathcal{A}}}
\newcommand{\mC}{{\mathcal{C}}}

\newcommand{\mU}{{\mathcal{U}}}
\newcommand{\mN}{{\mathcal{N}}}
\newcommand{\mM}{{\mathcal{M}}}
\newcommand{\mR}{{\mathcal{R}}}

\newcommand{\mD}{{\mathcal{D}}}
\newcommand{\mE}{{\mathcal{E}}}

\newcommand{\id}{{\mathbbm{1}}}

\newcommand{\vac}{{\mathrm{vac}}}
\newcommand{\rep}{{\mathrm{rep}}}

\newcommand{\scrS}{{\mathscr{S}}}

\newcommand{\scrX}{{\mathscr{X}}}

\newcommand{\vj}{{\gv{j}}}

\newcommand{\vK}{{\v{K}}}

\newcommand{\bR}{{\mathbb{R}}}
\newcommand{\bC}{{\mathbb{C}}}

\renewcommand{\Re}{{\mathrm{Re}}}

\newcommand{\opt}{{\mathrm{opt}}}
\newcommand{\cov}{{\mathrm{cov}}}
\newcommand{\supp}{{\mathrm{supp}}}
\renewcommand{\epsilon}{\varepsilon}
\newcommand{\appropto}{\mathrel{\vcenter{
  \offinterlineskip\halign{\hfil$##$\cr
    \propto\cr\noalign{\kern2pt}\sim\cr\noalign{\kern-2pt}}}}}

\newcommand{\LtoS}{{S\leftarrow L}}
\newcommand{\StoL}{{L\leftarrow S}}
\newcommand{\CtoSA}{{SA\leftarrow C}}
\newcommand{\SAtoC}{{C \leftarrow SA}}
\newcommand{\CtoLA}{{LA \leftarrow C}}
\newcommand{\LAtoC}{{C \leftarrow LA}}

\let\baraccent=\= 
\renewcommand{\=}[1]{\stackrel{#1}{=}} 

\newcommand{\thmref}[1]{\hyperref[#1]{Theorem~\ref{#1}}}
\newcommand{\lemmaref}[1]{\hyperref[#1]{Lemma~\ref{#1}}}
\newcommand{\figref}[1]{\hyperref[#1]{Fig.~\ref{#1}}}
\newcommand{\figaref}[1]{\hyperref[#1]{Fig.~\ref{#1}a}}
\newcommand{\figbref}[1]{\hyperref[#1]{Fig.~\ref{#1}b}}
\newcommand{\figcref}[1]{\hyperref[#1]{Fig.~\ref{#1}c}}
\newcommand{\figdref}[1]{\hyperref[#1]{Fig.~\ref{#1}d}}
\newcommand{\figeref}[1]{\hyperref[#1]{Fig.~\ref{#1}e}}
\renewcommand{\eqref}[1]{\hyperref[#1]{Eq.~(\ref{#1})}}
\newcommand{\secref}[1]{\hyperref[#1]{Sec.~\ref{#1}}}
\newcommand{\eqsref}[2]{\hyperref[#1]{Eqs.~(\ref{#1})-(\ref{#2})}}
\newcommand{\appref}[1]{\hyperref[#1]{Appx.~\ref{#1}}}

\usepackage{dsfont}

\newtheorem{theorem}{Theorem}
\newtheorem{lemma}{Lemma}

\title{New perspectives on covariant quantum error correction}

\author{Sisi Zhou}
\email{sisi.zhou26@gmail.com}
\affiliation{Department of Physics, Yale University, New Haven, Connecticut 06511, USA}
\affiliation{Pritzker School of Molecular Engineering, The University of Chicago, Illinois 60637, USA}

\author{Zi-Wen Liu}\email{zliu1@perimeterinstitute.ca}
\affiliation{Perimeter Institute for Theoretical Physics, Waterloo, Ontario N2L 2Y5, Canada}

\author{Liang Jiang}\email{liang.jiang@uchicago.edu}
\affiliation{Pritzker School of Molecular Engineering, The University of Chicago, Illinois 60637, USA}

\date{}

\begin{document}

\maketitle

\begin{abstract}
Covariant codes are quantum codes such that a symmetry transformation on the logical system could be realized by a symmetry transformation on the physical system, usually with limited capability of performing quantum error correction (an important case being the Eastin--Knill theorem). 
The need for understanding the limits of covariant quantum error correction arises in various realms of physics including fault-tolerant quantum computation, condensed matter physics and quantum gravity. Here, we explore covariant quantum error correction with respect to continuous symmetries from the perspectives of quantum metrology and quantum resource theory, establishing solid connections between these formerly disparate fields. We prove new and powerful lower bounds on the infidelity of covariant quantum error correction, which not only extend the scope of previous no-go results but also provide a substantial improvement over existing bounds. Explicit lower bounds are derived for both erasure and depolarizing noises. We also present a type of covariant codes which nearly saturates these lower bounds. 
\end{abstract}

\section{Introduction}
\label{sec:intro}

Quantum error correction (QEC) is a standard approach to protecting quantum systems against noises, 
which for example allows the possibility of practical quantum computing and has been a central research topic in quantum information~\cite{nielsen2002quantum,gottesman2010introduction,lidar2013quantum}. 
The key idea of QEC is to encode the logical state into a small code subspace in a large physical system and correct noises using the redundancy in the entire Hilbert space. 
As a result, the structure of noise must also place restrictions on  QEC codes.
This feature was beautifully captured by the Eastin--Knill theorem~\cite{eastin2009restrictions} (see also~\cite{bravyi2013classification,pastawski2015fault,jochym2018disjointness,wang2019quasi}), which states that any non-trivial local-error-correcting quantum code does not admit transversal implementations of a universal set of logical gates, ruling out the possibility of realizing fault-tolerant quantum computation using only transversal gates.

In particular, any finite-dimensional local-error-correcting quantum code only admits a finite number of transversal logical operations, which forbids the existence of codes covariant with continuous symmetries (discrete symmetries are allowed though~\cite{hayden2017error,faist2019continuous}). 
More generally, quantum codes under symmetry constraints, namely covariant codes, are of great practical and theoretical interest.
In general, a quantum code is covariant with respect to a logical Hamiltonian $H_L$ and a physical Hamiltonian $H_S$ if any symmetry transformation $e^{-iH_L\theta}$ is encoded into a symmetry transformation $e^{-iH_S\theta}$ in the physical system. 
Besides important implications to fault-tolerant quantum computation, covariant QEC is also closely connected to many other topics in quantum information and physics, such as quantum reference frames and quantum clocks~\cite{preskill2000quantum,hayden2017error,woods2019continuous}, symmetries in the AdS/CFT correspondence~\cite{almheiri2015bulk,pastawski2015holographic,harlow2018constraints,harlow2018symmetries,kohler2019toy,gschwendtner2019quantum,faist2019continuous,woods2019continuous} and approximate QEC in condensed matter physics~\cite{brandao2019quantum}. Although covariant codes cannot be perfectly local-error-correcting, they can still approximately correct errors with the infidelity depending on the number of subsystems, the dimension of each subsystem, etc. 
The quantifications of such infidelity in covariant QEC were explored recently, leading to an approximate, or robust, version of the Eastin--Knill theorem \cite{faist2019continuous,woods2019continuous}, using complementary channel techniques~\cite{beny2010general,hayden2008decoupling,beny2018approximate}. Note that these existing results only apply to erasure errors and random phase errors at unknown locations.


In this paper, we investigate covariant QEC from the perspectives of quantum metrology and quantum resource theory, which not only establishes conceptual and technical links between these seemingly separate fields, but also leads to a series of improved understandings and bounds on the performance of covariant QEC. 
Quantum metrology studies the ultimate limit on parameter estimation in quantum systems~\cite{giovannetti2011advances,degen2017quantum,braun2018quantum,pezze2018quantum,pirandola2018advances}. 
Covariant QEC is naturally a quantum metrological protocol---estimating the angle of any rotation of the physical system is equivalent to estimating that of the logical system with protection against noise. 
There is a no-go theorem in quantum metrology stating that perfectly error-correcting codes admitting a non-trivial logical Hamiltonian do not exist if the physical Hamiltonian falls into the Kraus span of the noise channel, which is known as the HKS condition~\cite{escher2011general,demkowicz2012elusive,demkowicz2014using,yuan2017quantum,demkowicz2017adaptive,zhou2018achieving,zhou2020theory}. 
It is also a sufficient condition of the non-existence of perfectly covariant QEC codes. 
When the HKS condition is satisfied, we establish a connection between the quantum Fisher information (QFI) of quantum channels~\cite{fujiwara2008fibre,demkowicz2014using,hayashi2011comparison,yuan2017fidelity,zhou2020theory,katariya2020geometric} and the performance (or infidelity) of covariant QEC, which gives rise to the desired lower bound.
We could also understand covariant QEC in terms of the resource theory of asymmetry~\cite{Gour_2008,MarvianSpekkens14,marvian2014extending} with respect to translations generated by Hamiltonians, where the covariant QEC procedures may naturally be represented by free operations.  
In quantum resource theory, we also have no-go theorems which dictate that pure resource states cannot be perfectly distilled from generic mixed states  \cite{FangLiu19:nogo,RBLT19,marvian2020coherence}, thereby ruling out the possibility of perfect covariant QEC.  By further analyzing suitable resource monotones, in particular a type of QFI \cite{marvian2020coherence}, we derive a lower bound on the infidelity of covariant QEC, which behaves similarly to the metrological bounds.

Our approaches and results on covariant QEC are innovative and also advantageous compared to previous ones in many ways. The bounds generalize the no-go theorems for covariant QEC from local Hamiltonians with erasure errors to generic Hamiltonian and noise structures. In the special case of erasure noise, our lower bounds improve the previous results in the small infidelity limit~\cite{faist2019continuous}. Furthermore, we    demonstrate that there is a type of covariant codes called thermodynamic codes~\cite{faist2019continuous,brandao2019quantum} that 
saturates the lower bound for erasure noise and matches the scaling of the lower bound for depolarizing noise, while previous bounds only apply to the erasure noise setting and were not known to be saturable~\cite{faist2019continuous}.

\section{Preliminaries: Covariant codes}
\label{sec:cov}

A quantum code is a subspace of a physical system $S$, usually defined by the image of an (usually isometric) encoding channel $\mE_{\LtoS}$ from a logical system $L$. We call a code $\mE_{\LtoS}$ \emph{covariant} if there exists a logical Hamiltonian $H_L$ and a physical Hamiltonian $H_S$ such that 
\begin{equation}
\mE_{\LtoS} \circ \mU_{L,\theta} = \mU_{S,\theta} \circ \mE_{\LtoS}, \;\forall \theta \in \bR, 
\end{equation}
where 
$\mU_{L,\theta}(\rho_{L}) = e^{-iH_{L}\theta} \rho_L e^{iH_{L}\theta}$ and $\mU_{S,\theta}(\rho_{S}) = e^{-iH_{S}\theta} \rho_S e^{iH_{S}\theta}$ are the symmetry transformations on the logical and physical systems, respectively.
We assume that the dimensions of the physical and logical systems $d_S$ and $d_L$ are both finite and $H_L$ is non-trivial ($H_L \not\propto \id$). For simplicity, we also assume all Hamiltonians in this paper are traceless unless stated otherwise, and we use $\Delta H_L$ and $\Delta H_{S}$ to denote the difference between the maximum and minimum eigenvalues of the operators.

We say a quantum code is \emph{error-correcting} under a noise channel $\mN_{S}$, if $\mN_{S}$ is invertible inside the code subspace, i.e., if there exists a CPTP map $\mR_{\StoL}$ such that
\begin{equation}
\mR_{\StoL} \circ \mN_{S} \circ  \mE_{\LtoS} = \id_L.
\end{equation} 
We assume that the output space of the noise channel $\mN_S$ is still $S$ for simplicity, although our results also apply to situations where the output system is different. The error-correcting property of a quantum code is often incompatible with its covariance with respect to continuous symmetries.  One representative example is the non-existence of error-correcting codes which can simultaneously correct non-trivial local errors and be covariant with respect to a local $H_S$~\cite{eastin2009restrictions,hayden2017error}. 
However, one may still consider approximate QEC with covariant codes~\cite{brandao2019quantum,faist2019continuous,woods2019continuous}. 
Then a question that naturally arises is how accurate covariant codes can be against certain noises.
To characterize the infidelity of an approximate QEC code, we use the worst-case entanglement fidelity $f(\Phi_1,\Phi_2)$ and the Choi entanglement fidelity~\cite{schumacher1996sending,gilchrist2005distance} defined by
\begin{equation}
f(\Phi_1,\Phi_2) = \min_{\rho}f((\Phi_1\otimes \id_R)(\rho),(\Phi_2\otimes \id_R)(\rho))
\end{equation} 
and 
\begin{equation}
f_\Choi(\Phi_1,\Phi_2) = f((\Phi_1\otimes \id_R)(\ket{\gamma}\bra{\gamma}),(\Phi_2\otimes \id_R)(\ket{\gamma}\bra{\gamma}))
\end{equation} 
for two quantum channels $\Phi_1$ and $\Phi_2$, where the fidelity between two states $\rho$ and $\sigma$ is given by \sloppy $f(\rho,\sigma) = \trace(\sqrt{\rho^{1/2}\sigma\rho^{1/2}})$~\cite{nielsen2002quantum}, and $R$ is a reference system identical to the system $\Phi_{1,2}$ acts on (which we assume to be $L$) and the maximally entangled state $\ket{\gamma} = \frac{1}{\sqrt{d_L}}\sum_i \ket{i}_L\ket{i}_R$.  After optimizing over recovery channels $\mR_{\StoL}$, we may define the \emph{infidelity} and the \emph{Choi infidelity}   of a code $\mE_{\LtoS}$  
respectively by\footnote{There are other equivalent definitions of the code infidelity in the literature that have a quadratic difference in terms of scaling with ours, e.g., $\sqrt{1-f^2}$ in \cite{faist2019continuous} or $\sqrt{1-f}$ in \cite{kubica2020using}.}
\begin{equation}
\epsilon(\mN_{S},\mE_{\LtoS}) = 1 - \max_{\mR_{\StoL}} f^2(\mR_{\StoL} \circ \mN_{S} \circ  \mE_{\LtoS},\id_L), 
\end{equation} 
and 
\begin{equation}
\epsilon_\Choi(\mN_{S},\mE_{\LtoS}) = 1 - \max_{\mR_{\StoL}} f_\Choi^2(\mR_{\StoL} \circ \mN_{S} \circ  \mE_{\LtoS},\id_L). 
\end{equation} 
Note that the Choi infidelity reflects the ``average-case'' behavior in the sense that it is directly related to $\epsilon_\avg = 1 - \max_{\mR_{\StoL}}\int d\psi \bra{\psi}\mR_{\StoL}\circ \mN_S \circ \mE_{\LtoS}(\ket{\psi}\bra{\psi})\ket{\psi}$, where the integral is over the Haar measure,  by $\epsilon_\Choi = \frac{d_L+1}{d_L}\epsilon_\avg$.
We will sometimes simply use $\epsilon$ and $\epsilon_\Choi$ to denote $\epsilon(\mN_{S},\mE_{\LtoS})$ and $\epsilon_\Choi(\mN_{S},\mE_{\LtoS})$ when the system under consideration is unambiguous. Clearly, $\epsilon \geq \epsilon_\Choi$.  We will use $\mR^{\opt}_{\StoL}$ to represent the optimal recovery channel achieving $\epsilon$ and $\mI_{L}$ to denote the effective noise channel $\mR_{\StoL} \circ \mN_{S} \circ \mE_{\LtoS}$ in the logical system. 

\section{Metrological bound}
\label{sec:metro}

Recently, QEC emerges as a useful tool to enhance the sensitivity of an unknown parameter in quantum metrology~\cite{kessler2014quantum,arrad2014increasing,dur2014improved,lu2015robust,reiter2017dissipative,sekatski2017quantum,demkowicz2017adaptive,zhou2018achieving,kapourniotis2019fault,layden2018spatial,layden2019ancilla,zhou2019optimal,zhou2020theory}. 
A good approximately error-correcting covariant code naturally provides a good quantum sensor to estimate an unknown parameter $\theta$ in the symmetry transformation $e^{-i H_S \theta}$. 
Consider a quantum signal $e^{-i H_S \theta}$ in the physical system, for example, the magnetic field in a spin system with $H_S$ being the angular momentum operator. The optimal sensitivity is usually limited by the strength of noise in the system. Instead of using the entire system to probe the signal, one could prepare an encoded probe state using covariant codes where $H_S$ is mapped to $H_L$ associated with the logical system. Covariant codes with low infidelity significantly reduce the noise in the logical system and therefore provide a good sensitivity of the signal. 

No-go theorems in quantum metrology~\cite{escher2011general,demkowicz2012elusive,demkowicz2014using,yuan2017quantum,demkowicz2017adaptive,zhou2018achieving,zhou2020theory} prevent the existence of perfectly error-correcting covariant codes in the above scenario. In particular, it was known that given a noise channel $\mN_S(\cdot) = \sum_{i=1}^r K_{S,i} (\cdot) K_{S,i}^\dagger$ and a physical Hamiltonian $e^{-iH_S \theta}$, there exists an encoding channel $\mE_{\LtoS}$ and a recovery channel $\mR_{\StoL}$ such that 
\begin{equation}
\label{eq:RNUE}
\mR_{\StoL} \circ \mN_S \circ \mU_{S,\theta} \circ \mE_{\LtoS}
\end{equation}
is a non-trivial unitary channel 
only if 
$H_S \not\in {\rm span}\{K_{S,i}^\dagger K_{S,j},\forall i,j\}$~\cite{zhou2020theory}. However, the above channel (\eqref{eq:RNUE}) with respect to any perfectly error-correcting covariant code is simply $\mU_{L,\theta}$. Therefore, we conclude that perfectly error-correcting covariant codes do not exist when
\begin{equation}
\label{eq:HKS}
H_S \in {\rm span}\{K_{S,i}^\dagger K_{S,j},\forall i,j\},
\end{equation}
which we call the ``Hamiltonian-in-Kraus-span'' (HKS) condition. One could check that local Hamiltonians with non-trivial local errors is a special case of the HKS condition. 

Note that the no-go result might be circumvented when the system dimension is infinite. To be more specific, perfect error-correcting codes that are covariant under $U(1)$ or even an arbitrary group $G$ can be constructed using quantum systems that transform as the regular representation of $G$ which are infinite-dimensional when $G$ is infinite and are equivalent to the notions of idealized clocks or perfect reference frames~\cite{hayden2017error,faist2019continuous,woods2019continuous}. 

\subsection{Quantum channel estimation}

From the discussion above, we saw that no-go theorems in quantum metrology help us extend the scope of the Eastin--Knill theorem for covariant codes. As we will see below, a powerful lower bound for the infidelity of covariant codes could also be derived thanks to recent developments in quantum channel estimation~\cite{demkowicz2012elusive,demkowicz2014using,zhou2020theory,fujiwara2008fibre,hayashi2011comparison,yuan2017fidelity,katariya2020geometric}. 

Here we first review the definitions of QFIs for quantum states and then introduce the extensions to quantum channels. The QFI is a good measure of the amount of information a quantum state $\rho_\theta$ carries about an unknown parameter $\theta$, characterized by the the quantum Cram\'{e}r-Rao bound~\cite{helstrom1976quantum,holevo2011probabilistic,paris2009quantum,braunstein1994statistical}, $
\delta \theta \geq 1/{\sqrt{N_{\rm expr} F(\rho_\theta)}}$,
where $\delta \theta$ is the standard deviation of any unbiased estimator of $\theta$, $N_{\rm expr}$ is the number of repeated experiments and $F(\rho_\theta)$ is the QFI of $\rho_\theta$. The QFI as the quantum generalization of the classical Fisher information is not unique, due to the noncommutativity of quantum operators. Two most commonly used QFIs are the symmetric logarithmic derivative (SLD) QFI and the right logarithmic derivative (RLD) QFI, respectively defined by~\cite{helstrom1976quantum,holevo2011probabilistic,Yuen1973multiple}, 
\begin{gather}
F_\txs(\rho_\theta) = \trace(\rho_\theta (L^\txs_\theta)^2 ),\quad \partial_\theta \rho_\theta = \frac{1}{2}(L^\txs_\theta \rho_\theta + \rho_\theta L^\txs_\theta),\\
F_\txr(\rho_\theta) = \trace( \rho_\theta L_\theta^\txr L_\theta^{\txr\dagger} ),\quad \partial_\theta \rho_\theta = \rho_\theta L_\theta^\txr,
\end{gather}
 where the SLD $L^\txs_\theta$ is Hermitian and the RLD $L^\txr_\theta$ is linear. Note that $F_\txr(\rho_\theta) = +\infty$ if ${\supp}(\partial_\theta \rho_\theta) \not\subseteq {\supp}(\rho_\theta)$. 
The QFIs satisfy many nice information-theoretic properties~\cite{katariya2020geometric}, such as additivity $F(\rho_\theta\otimes\sigma_\theta) = F(\rho_\theta) + F(\sigma_\theta)$ and monotonicity $F(\mN(\rho_\theta)) \leq F(\rho_\theta)$ for $\theta$-independent  channel $\mN$. 
Note that the SLD QFI is the smallest monotonic quantum extension from the classical Fisher information and the quantum Cram\'{e}r-Rao bound with respect to the SLD QFI is saturable asymptotically ($N_{\rm expr} \gg 1$). 

In this section, we will focus on the SLD QFI for quantum channels. Discussions on the RLD QFI for quantum channels will be delayed to \secref{sec:res} where it is used. Given a quantum channel $\mN_\theta$, the (entanglement-assisted) SLD QFI of $\mN_\theta$~\cite{fujiwara2008fibre} is defined by 
\begin{equation}
F_\txs(\mN_\theta) = \max_{\rho} F_\txs((\mN_\theta\otimes \id_R)(\rho)),
\end{equation}
where $R$ is an unbounded reference system. 
The regularized SLD QFI for quantum channels also has a single-letter expression~\cite{zhou2020theory}:
\begin{equation}
\label{eq:SLD-def}
\begin{split}
&F_\txs^\reg(\mN_\theta)=\lim_{N\rightarrow \infty}\frac{F_\txs(\mN_\theta^{\otimes N})}{N} =
\begin{cases}
4\min_{h:\beta_\theta=0}\norm{\alpha_\theta} & \text{(S)},\\
+\infty & \text{otherwise},\\
\end{cases}\\
&\;\;\;\text{(S):~~} i  \sum_{i=1}^r K_{i,\theta}^\dagger\ptheta K_{i,\theta} \in {\rm span}\{K_{i,\theta}^\dagger K_{j,\theta},\forall i,j\}, 
\end{split}
\end{equation}
where $\mN_{\theta}(\cdot) = \sum_{i=1}^r K_{i,\theta}(\cdot) K_{i,\theta}^\dagger$, $h$ is a Hermitian operator in $\bC^{r \times r}$, $\norm{\cdot}$ is the operator norm and 
\begin{align}
\alpha_\theta &= (\ptheta \vK_\theta + i h \vK_\theta)^\dagger(\ptheta \vK_\theta + i h \vK_\theta),\\
\beta_\theta &= \vK_\theta^\dagger h\vK_\theta - i  \vK_\theta^\dagger\ptheta\vK_\theta. 
\end{align}
\sloppy Here $\vK_\theta^T = (K_{1,\theta}^T~K_{2,\theta}^T~\cdots~K_{r,\theta}^T) \in \bC^{d \times r d}$ is a block matrix where $^T$ means the transpose of a matrix. 

Note that when (S) is violated, $F^{\reg}_\txs(\mN_\theta) = \infty$ because we will have $F_\txs(\mN^{\otimes N}_\theta) \propto N^2$~\cite{zhou2020theory}.
The regularized SLD QFI is additive (see \appref{app:additivity}) and could be calculated efficiently using semidefinite programs (SDP)~\cite{demkowicz2012elusive}. 
It is also monotonic, satisfying 
$
F_\txs^\reg(\Phi_1\circ (\mN_{\theta}\otimes \id) \circ \Phi_2) \leq F_\txs^\reg(\mN_{\theta})$ 
where $\Phi_{1,2}$ are any parameter-independent channels, due to the monotonicity of the state QFI.

\subsection{Metrological bound}

In order to derive a lower bound on the infidelity of covariant codes using the channel QFI, we note that the channel QFI provides an upper limit to the sensitivity of $\theta$ for $\mN_{S,\theta} = \mN_S \circ \mU_{S,\theta}$, which cannot be broken using covariant QEC. 

Concretely, we consider an encoding scheme based on covariant QEC, where the original system consists of the physical system $S$ and a noiseless ancillary qubit $A$ and the logical system is a two-dimensional system $C$. 
Suppose the original system is subject to Hamiltonian evolution $e^{-iH_S\theta}$ and noise channel $\mN_S$, then the logical system will be subject to a $Z$-rotation signal and a rotated dephasing noise (defined later). Specifically, the Hamiltonian evolution in $C$ is $e^{-i(\Delta H_L)Z_C\theta/2}$ where $Z_C$ is the Pauli-Z operator and the rotated dephasing noise channel has a noise rate smaller than the code infidelity.
The monotonicity of the channel QFI guarantees that the QFI of the logical system cannot surpass the QFI of the original system and thus leads to a lower bound on the code infidelity. The Hamiltonians and noises before and after the covariant encoding scheme are listed in \figcref{fig:dephasing}.

To roughly estimate the scaling of the lower bound in the small infidelity limit, consider $N$ logical qubits each under a unitary evolution $e^{-i(\Delta H_L)Z_C\theta/2}$ with a noise rate $\epsilon$. It is known that the SLD QFI of a noiseless $N$-qubit GHZ state is $(\Delta H_L)^2 N^2$~\cite{giovannetti2006quantum}. Taking $N = \Theta(1/\epsilon)$, the total noise can be bounded by a small constant, and the state SLD QFI per qubit is still roughly $\Theta((\Delta H_L)^2 N) = \Theta((\Delta H_L)^2/\epsilon)$, which is always no greater than the regularized channel SLD QFI $F^\reg_\txs(\mN_{S,\theta})$ before QEC. Thus, $\epsilon$ must be lower bounded by $\Theta((\Delta H_L)^2/F^\reg_\txs(\mN_{S,\theta}))$. In fact, using the regularized SLD QFI for quantum channels, we can prove the following theorem: 
\begin{theorem}
\label{thm:SLD}
Consider a covariant code $\mE_{\LtoS}$ under a noise channel $\mN_S(\cdot) = \sum_{i=1}^r K_{S,i} (\cdot) K_{S,i}^\dagger$. If the HKS condition is satisfied, i.e.,
\begin{equation}
\label{eq:cond-1}
H_S \in {\rm span}\{K_{S,i}^\dagger K_{S,j},\forall i,j\},
\end{equation}
then $\epsilon(\mN_S,\mE_{\LtoS})$ is lower bounded as follows, 
\begin{equation}
\label{eq:bound-S}
\epsilon \geq \ell_1\left(\frac{(\Delta H_L)^2}{4 F_\txs^\reg(\mN_S,H_S)}\right) = \frac{(\Delta H_L)^2}{4 F_\txs^\reg(\mN_S,H_S)} + O\left(\left(\frac{(\Delta H_L)^2}{4 F_\txs^\reg(\mN_S,H_S)}\right)^2\right),
\end{equation}
where $\ell_1(x) = (1 + 4 x - \sqrt{1 + 4 x})/(2 (1 + 4 x)) = x + O(x^2)$ is a monotonically increasing function and $F_\txs^\reg(\mN_S,H_S)$ is the regularized SLD QFI of $\mN_{S,\theta}$. 
Specifically, $F_\txs^\reg(\mN_S,H_S) = 4 \min_{h:\beta_S=0}\norm{\alpha_S}$, $h$ is a Hermitian operator in $\bC^{r\times r}$. $\alpha_S$ and $\beta_S$ are Hermitian operators acting on $S$ defined by 
$\alpha_S 
= \vK_S^\dagger h^2\vK_S - H_S^2$ and $\beta_S = \vK_S^\dagger h \vK_S - H_S$, 
where 
$\vK^T = (K_1^T ~ K_2^T ~ \cdots ~ K_r^T) \in \bC^{d_S \times r d_S}$ is a block matrix.
\end{theorem}

Note that $\ell_1(x) \leq 1/2$ for all $x$ and the lower bound does provide useful information when $\epsilon \geq 1/2$. 
We remark that, in principle, the regularized SLD QFI on the right-hand side of \eqref{eq:bound-S} could be replaced by any other type of QFIs because the SLD QFI is the smallest monotonic quantum extension from the classical Fisher information.  However, the regularized SLD QFI will always lead to the tightest bound using our metrological approach, as we will see later. 

\subsection{Proof of \texorpdfstring{Theorem~\ref{thm:SLD}}{Theorem 1}} 

The main obstacle to proving \thmref{thm:SLD} is to relate the infidelity of covariant codes to the QFI of the effective quantum channel in the logical system. Here we overcome this obstacle and provide a proof of \thmref{thm:SLD} by employing entanglement-assisted QEC to reduce $\mN_{S,\theta}$ to rotated dephasing channels (the composition of dephasing channels and Pauli-Z rotations) whose QFI has simple mathematical forms, and then connecting the noise rate of the rotated dephasing channels to the infidelity of the covariant codes (see \figref{fig:dephasing}).

\begin{figure}[tb]
\centering
\includegraphics[width=0.95\textwidth]{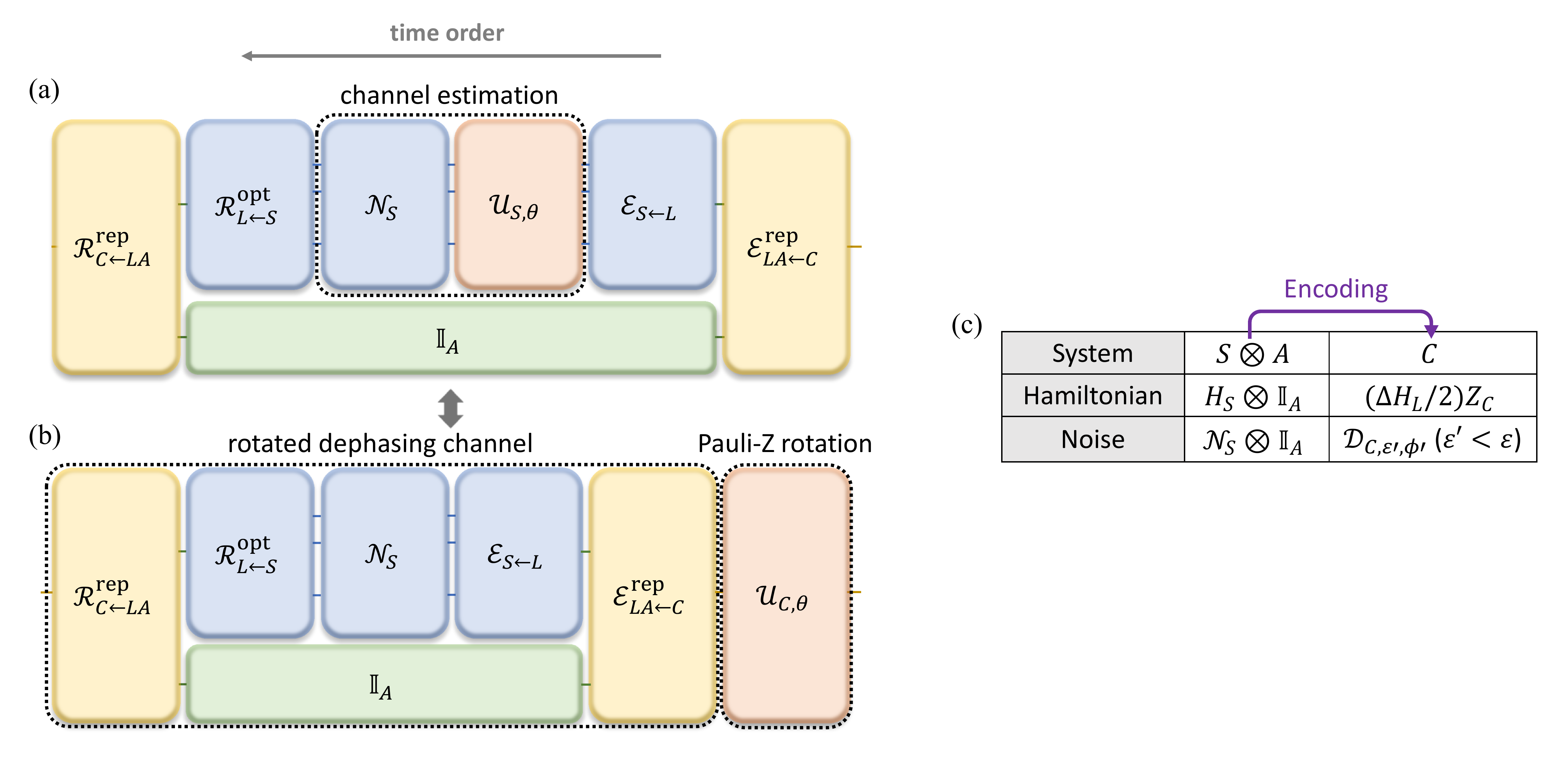}
\caption{\label{fig:dephasing} 
Reduction of $\mN_{S,\theta} = \mN_{S} \circ \mU_{S,\theta}$ to rotated dephasing channels using entanglement-assisted QEC. (a) represents the quantum channel $\mR_{\SAtoC}\circ (\mN_{S,\theta} \otimes \id_A) \circ \mE_{\CtoSA}$ with a channel QFI no larger than $F(\mN_{S,\theta})$. Because of  the covariance of the code, (a) is equivalent to (b) which consists of a Pauli-Z rotation $\mU_{C,\theta}$ and a $\theta$-independent rotated dephasing channel $\mI_C$ whose noise rate is smaller than $\epsilon(\mN_S,\mE_{\LtoS})$ (see \lemmaref{lemma:dephasing}). The Hamiltonians and noises before and after the encoding are listed in (c).   
} 
\end{figure}

Single-qubit rotated dephasing channels take the form
\begin{equation}
\label{eq:dephasing}
\mD_{p,\phi}(\rho) = (1-p) e^{-i\frac{\phi}{2}Z} \rho e^{i\frac{\phi}{2}Z} + p e^{-i\frac{\phi}{2}Z} Z \rho Z e^{i\frac{\phi}{2}Z}, 
\end{equation}
where $Z$ is the Pauli-Z operator, $0 < p < 1/2$ and $\phi$ is real. 
When $\phi$ is a function of $\theta$, we could calculate the regularized SLD QFI of $\mD_{p,\phi_\theta}$ (see Appx. B in \cite{zhou2020theory} or \cite{kolodynski2013efficient}):
\begin{align}
F^\reg_\txs(\mD_{p,\phi_\theta}) 
= \frac{(1-2p)^2 (\ptheta \phi_\theta)^2}{4p(1-p)},
\end{align} 
which are both inversely proportional to the noise rate $p$ when $p$ is small---a crucial feature in deriving the lower bounds. 

Next, we present an entanglement-assisted QEC protocol to reduce $\mN_S$ to rotated dephasing channels with a noise rate lower than $\epsilon(\mN_S,\mE_{\LtoS})$. Let $\ket{0_L}$ and $\ket{1_L}$ be eigenstates respectively corresponding to the largest and the smallest eigenvalues of $H_L$. Consider the following two-dimensional entanglement-assisted code 
\begin{equation}
\mE^{\rep}_{\CtoLA}(\ket{0_{C}}) = \ket{0_L0_A}, \quad \mE^{\rep}_{\CtoLA}(\ket{1_{C}}) = \ket{1_L1_A},
\end{equation} 
where $A$ is a noiseless ancillary qubit and the superscript $\rep$ means ``repetition''. The encoding channel from the two-level system $C$ to $SA$ will simply be $\mE_{\CtoSA} = \big(\mE_{\LtoS} \otimes \id_A\big) \circ \mE^{\rep}_{\CtoLA}$. $\mE_{\CtoSA}$ is still a covariant code whose logical and physical Hamiltonians are
\begin{equation}
H_{C} = \frac{\Delta H_L}{2} \cdot Z_{C},\quad H_{SA} = H_S \otimes \id_A.
\end{equation}
The noiseless ancillary qubit will help us suppress off-diagonal noises in the system because any single qubit bit-flip noise on $L$ could be fully corrected by mapping $\ket{i_L j_A}$ to $\ket{j_C}$ for all $i,j$. In fact, $\mN_S$ will be reduced to a rotated dephasing channel, as shown in the following lemma: 
\begin{lemma}
\label{lemma:dephasing}
Consider a noise channel $\mN_{SA} = \mN_S\otimes \id_A$. There exists a recovery channel $\mR_{\SAtoC}$ such that the effective noise channel $\mI_{C} = \mR_{\SAtoC} \circ \mN_{SA} \circ \mE_{\CtoSA}$ is a rotated dephasing channel, satisfying 
\begin{equation}
\label{eq:logical-dephasing}
\mI_{C} = \mD_{C,\epsilon',\phi'}, 
\end{equation}
where $\epsilon' \leq \epsilon(\mN_S,\mE_{\LtoS})$. 
\end{lemma}
\begin{proof}
Consider the following recovery channel 
\begin{equation}
\mR_{\SAtoC}  = \mR^{\rep}_{\LAtoC} \circ \big(\mR^{\opt}_{\StoL}\otimes \id_A\big),
\end{equation}
where 
\begin{equation}
\begin{split}
\mR^{\rep}_{\LAtoC}(\cdot)  = 
&(\ket{0_C}\bra{0_L0_A}+\ket{1_C}\bra{1_L1_A})(\cdot)(\ket{0_L0_A}\bra{0_C}+\ket{1_L1_A}\bra{1_C}) + \\
& (\ket{0_C}\bra{1_L0_A}+\ket{1_C}\bra{0_L1_A})(\cdot)(\ket{1_L0_A}\bra{0_C}+\ket{0_L1_A}\bra{1_C}) + \\
& \sum_{i=2}^{d_L-1} (\ket{0_C}\bra{i_L0_A}+\ket{1_C}\bra{i_L1_A})(\cdot)(\ket{i_L0_A}\bra{0_C}+\ket{i_L1_A}\bra{1_C}). 
\end{split}
\end{equation}
The last term above disappears when $d_L = 2$. 

One could check that 
\begin{equation}
\mI_C(\ket{k_C}\bra{j_C}) = 
\begin{cases}
\ket{k_C}\bra{j_C} , & k=j,\\ 
(1-2\epsilon')e^{i\phi'(k-j)} \ket{k_C}\bra{j_C}, & k\neq j, 
\end{cases}
\end{equation}
which indicates that $\mI_{C} =\mD_{C,\epsilon',\phi'}$ (\eqref{eq:logical-dephasing}). 
Here,
\begin{equation}
\epsilon' \leq 1 - f^2(\mI_C,\id_C) \leq 1 - f^2(\mI_L^{\opt},\id_L) = \epsilon(\mN_S,\mE_{\LtoS}). 
\end{equation}
where $\mI_{L} = \mR^{\opt}_{\StoL} \circ \mN_{S} \circ \mE_{\LtoS}$ the first inequality follows from the worst-case entanglement fidelity for rotated dephasing channels (see \appref{app:dephasing-fidelity}), and the the second inequality follows from $\id_{C} = \mR^\rep_{\LAtoC} \circ \mE^\rep_{\CtoLA}$ and the monotonicity of the fidelity~\cite{nielsen2002quantum}. 
\end{proof}

\lemmaref{lemma:dephasing} shows that $\mN_S$ could be reduced to a rotated dephasing channel $\mI_C$ through entanglement-assisted QEC. Consider parameter estimation of $\theta$ in the quantum channel $\mN_{S,\theta} = \mN_S \circ \mU_{S,\theta}$. We have the error-corrected quantum channel 
\begin{equation}
\mN_{C,\theta} 
= \mR_{\SAtoC} \circ \big(\mN_{S,\theta} \otimes \id_A\big) \circ \mE_{\CtoSA}\\
= \mI_C \circ \mU_{C,\theta},
\end{equation}
equal to a rotated dephasing channel with noise rate $\epsilon'$ and phase $\phi_\theta = \phi' + \Delta H_L \theta$.
The monotonicity of the regularized channel SLD QFI implies that 
\begin{equation}
\label{eq:monotonic-QFI}
F_\txs^\reg(\mN_{S,\theta}) \geq F_\txs^\reg(\mN_{C,\theta}),
\end{equation}
where 
\begin{equation}
F_\txs^\reg(\mN_{S,\theta}) = 
\begin{cases}
F_\txs^\reg(\mN_{S},H_S) & H_S \in {\rm span}\{K_{S,i}^\dagger K_{S,j},\forall i,j\},\\
+\infty & \text{otherwise},\\
\end{cases}
\end{equation}
and 
\begin{equation}
F^\reg_\txs(\mN_{C,\theta}) = \frac{(1-2\epsilon')^2(\Delta H_L)^2}{4\epsilon'(1-\epsilon')}.
\end{equation}
Using \eqref{eq:monotonic-QFI} and $\epsilon' \leq \epsilon < 1/2$, we have 
\begin{equation}
\epsilon \cdot \frac{1 - \epsilon}{(1 - 2\epsilon)^2} \geq \frac{(\Delta H_L)^2}{4 F_\txs^\reg(\mN_S,H_S)}.
\end{equation}
\thmref{thm:SLD} then follows from the fact that $\ell_1(\cdot)$ is the inverse function of $x = \frac{\epsilon(1 - \epsilon)}{(1 - 2\epsilon)^2}$ for $\epsilon \in [0,1/2)$. Moreover, any other types of regularized monotonic QFIs of $\mD_{p,\phi_\theta}$ have the same values as the regularized SLD QFI of $\mD_{p,\phi_\theta}$~\cite{kolodynski2013efficient}. It implies that although the definition of QFI by generalizing the classical Fisher information is not unique, one cannot derive tighter bounds on the code infidelity by replacing the regularized SLD QFI with other types of QFIs in this proof.

\section{Resource-theoretic bound}
\label{sec:res}

Now we demonstrate how quantum resource theory provides another pathway towards characterizing the limitations of covariant QEC. 
More specifically, the covariance property of the allowed operations indicates close connections to the (highly relevant) resource theories of asymmetry, reference frames, coherence, and quantum clocks \cite{Gour_2008,MarvianSpekkens14,marvian2014extending,CoherenceRMP,marvian2020coherence}.  In our current context, we work with a resource theory of coherence (see e.g., \cite{marvian2020coherence} for more discussions on the setting) where the free (incoherent) states are those with density operators commuting with the physical Hamiltonian $H_S$, and the free operations are 
covariant operations $\mC_{\StoL}$ from $S$ to $L$ satisfying   
\begin{equation}
\mC_{\StoL} \circ \mU_{S,\theta} = \mU_{L,\theta} \circ \mC_{\StoL}, \;\forall \theta \in \bR. 
\end{equation}
The free (covariant) operations are incoherence-preserving, i.e., map incoherent states to incoherent states even with the assistance of reference systems~\cite{marvian2020coherence}. (Note that in \cite{marvian2020coherence} covariant operations are called time-translation invariant operations.)

The following lemma shows that the recovery channel $\mR_{\StoL}$ for a covariant code can be assumed to be covariant under two conditions: (1) the noise channel and the symmetry transformation commutes (e.g., satisfied by the erasure and depolarizing channels of interest here); (2) $U_S(\theta) = e^{-i H_{S}\theta}$ and $U_L(\theta) = e^{-i H_{L}\theta}$ share a common period $\tau$ (which is not necessarily the fundamental period), i.e., $U_{S,L}(\tau) = \id_{S,L}$. For example, when $U_S(\theta)$ and $U_L(\theta)$ are both representations of $U(1)$, $\tau = 2\pi$ is the common period, which is a standard assumption in the theory of quantum clocks~\cite{woods2019continuous,marvian2020coherence}. 
\begin{lemma}
\label{thm:cov}
Suppose $\mN_S \circ \mU_{S,\theta} = \mU_{S,\theta} \circ \mN_S$ and $U_{L,S}(\theta)$ share a common period $\tau$. Then the code infidelity $\epsilon(\mN_S,\mE_{\LtoS})$ and the Choi code infidelity $\epsilon_\Choi(\mN_S,\mE_{\LtoS})$ stay the same if the recovery channels are restricted to be covariant operations.  
\end{lemma}
\begin{proof}
Let $\mR^\opt_{\StoL}$ be the recovery channel such that $1 - f^2(\mR^\opt_{\LtoS}\circ \mN_S \circ \mE_{\LtoS},\id_L) = \epsilon(\mN_S,\mE_{\LtoS})$. Consider the following recovery channel: 
\begin{equation}
\label{eq:def-cov-rec}
\mR_{\StoL}^\cov = \frac{1}{\tau} \int_0^\tau  d\theta\;\mU_{L,\theta} \circ \mR^\opt_{\StoL} \circ \mU_{S,\theta}^\dagger . 
\end{equation}
We first observe that $\mR_{\StoL}^\cov$ is covariant:
\begin{equation}
\mR_{\StoL}^\cov \circ \mU_{S,\theta'} = \frac{1}{\tau}\int_0^\tau d\theta\; \mU_{L,\theta} \circ \mR^\opt_{\StoL} \circ \mU_{S,\theta-\theta'}^\dagger = \mU_{L,\theta'} \circ \mR_{\StoL}^\cov. 
\end{equation}
Furthermore,
\begin{equation}
\mR_{\StoL}^\cov \circ \mN_S \circ \mE_{\LtoS} = \frac{1}{\tau}\!\int_0^\tau\! d\theta\;\mU_{L,\theta} \circ (\mR^\opt_{\StoL} \circ \mN_S \circ \mE_{\LtoS}) \circ \mU_{L,\theta}^\dagger,
\end{equation}
Using the concavity of $f^2(\Phi,\id)$ with respect to $\Phi$ and the monotonicity of the worst-case fidelity~\cite{schumacher1996sending}, we have $1 - f^2(\mR^\cov_{\LtoS}\circ \mN_S \circ \mE_{\LtoS},\id_L) \leq \epsilon(\mN_S,\mE_{\LtoS})$. By definition, the equality holds. 

Similarly, one could construct an optimal covariant recovery channel with respect to the Choi code infidelity by replacing $\mR^\opt_{\StoL}$ with the optimal recovery channel achieving the minimum Choi infidelity in \eqref{eq:def-cov-rec} and prove its optimality by noting that $f^2(\Phi,\id)$ is linear with respect to the channel $\Phi$ and $f(\Phi,\id) = f(\mU\circ\Phi\circ\mU^\dagger,\id)$ for any unitary channel $\mU$. 
\end{proof}

\begin{figure}[tb]
\centering
\includegraphics[width=0.7\textwidth]{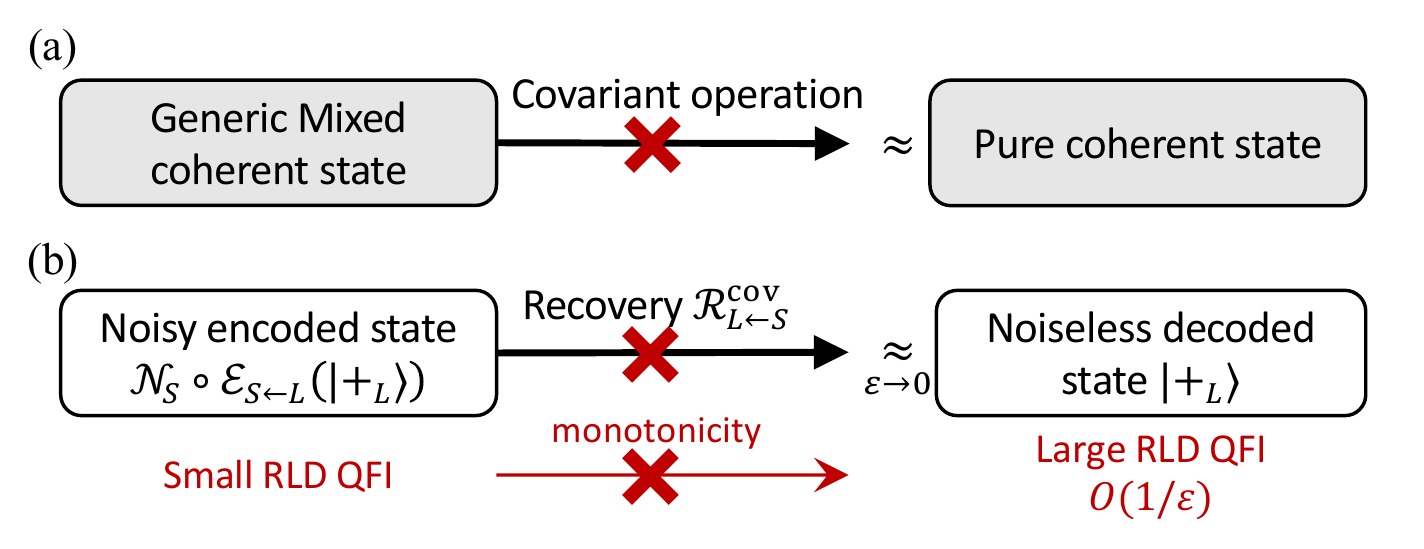}
\caption{\label{fig:coh} 
(a) As a case of the general no-go theorems for quantum resource purification, generic noisy coherent states cannot be transformed to states within a neighborhood of pure coherent states using covariant operations. 
(b) Specifically for covariant quantum error correction, a noisy coherent encoded state (e.g.~$\mN_S\circ \mE_{\LtoS}(\ket{+_L})$)
cannot be transformed to states within a neighborhood of a pure coherent logical state (e.g.~$\ket{+_L}$), that is, there exist fundamental limits on the overall accuracy of recovery (decoding). As illustrated at the bottom, our specific bounds on the accuracy are derived from the monotonicity of RLD QFI under covariant operations. The key intuition is that RLD QFI diverges as the infidelity with the target pure coherent state $\epsilon$ tends to zero,
thus recovery with too small $\epsilon$ is ruled out by the monotonicity of RLD QFI. 
} 
\end{figure}

This lemma allows us to formulate covariant QEC as a resource state conversion task, where one aims to transform noisy physical states to logical states by covariant operations (see \figaref{fig:coh}).  It has recently been found that, by analyzing suitable resource monotones (functions of states that are nonincreasing under free operations), one can prove strong lower bounds on the infidelity of transforming generic noisy states to pure resource states by any free operation, which underlies the important task of distillation (see \cite{FangLiu19:nogo,FL20,regula2020oneshot} for general results that apply to any well-behaved resource theory, and \cite{marvian2020coherence} for discussions specific to covariant operations).  Here, we use the RLD QFI for quantum states, which is  studied as a coherence monotone in \cite{marvian2020coherence}, to derive bounds on the performance of covariant QEC. 
In particular, the RLD QFI satisfies
\begin{equation}
\label{eq:monotone-RLD}
F_\txr(\mR^\cov_{\StoL}(\rho_S),H_L) \leq F_\txr(\rho_S,H_S),
\end{equation}
for all $\rho_S$ and covariant operations $\mR^\cov_{\StoL}$ where 
\begin{equation}
F_\txr(\rho,H) = F_\txr(e^{-iH\theta}\rho e^{iH\theta}) = 
\begin{cases}
\trace(H\rho^2 H \rho^{-1}) - \trace(\rho H^2) & \supp(H\rho H)\subseteq \supp(\rho),\\
+\infty & \text{otherwise.}
\end{cases}
\end{equation}
Notice that the RLD QFI $F_\txr(\rho,H)$ approaches infinity when $\rho$ is coherent and its purity $\trace(\rho^2)$ approaches one. 
Let $\rho_S = \mN_{S}\circ\mE_{\LtoS}(\rho_L)$ where $\rho_L$ is a pure coherent state. If the right-hand side of \eqref{eq:monotone-RLD} is finite, the left-hand side of \eqref{eq:monotone-RLD} must also be finite, and thus a perfectly error-correcting recovery channel does not exist (see \figbref{fig:coh}). That is, the RLD QFI is a distinguished coherence monotone which can rule out generic noisy-to-pure transformations and further induce lower bounds on the code infidelity. In fact, the RLD QFI $F_\txr(\rho,H)$ is lower bounded by $O(1/\epsilon)$ when $\rho$ is $\epsilon$-close to a pure state $\ket{\Psi}$ in terms of infidelity, characterized by the following lemma: 
\begin{lemma}[\cite{marvian2020coherence}]
\label{lemma:pure-rld}
Let $\teps = 1 - \braket{\Psi|\rho|\Psi}$. 
\begin{equation}
\label{eq:pure-rld}
F_\txr(\rho,H) \geq \left(V_{H}(\Psi) - \frac{3\sqrt{2\teps}(\Delta H)^2}{2}\right)\cdot \frac{1 - 3\teps + \teps^2}{\teps},
\end{equation}
where $V_H(\Psi) = \braket{\Psi|H^2|\Psi} - (\braket{\Psi|H|\Psi})^2$. 
\end{lemma}
\begin{proof}
According to Corollary 1 in Supplementary Note 3 in~\cite{marvian2020coherence}, 
\begin{equation}
\label{eq:pure-rld-1}
F_\txr(\rho,H) \geq V_{H}(\tPsi)\cdot \frac{1 - 3\teps + \teps^2}{\teps},
\end{equation}
where $\tPsi$ is the eigenvector of $\rho$ which corresponds to the largest eigenvalues of $\rho$. Let $\lambda_{\max}$ be the largest eigenvalue of $\rho$. 
Then $\braket{\Psi|\rho|\Psi} = 1 - \teps \leq \lambda_{\max}$. Moreover,
\begin{equation}
\begin{split}
1 - \teps &= \braket{\Psi|\rho|\Psi} \leq \bra{\Psi}(\lambda_{\max}\ket{\tPsi}\bra{\tPsi} + \rho - \lambda_{\max}\ket{\tPsi}\bra{\tPsi})\ket{\Psi}\\
&\leq \lambda_{\max}\abs{\braket{\Psi|\tPsi}}^2 + (1-\lambda_{\max})(1 - \abs{\braket{\Psi|\tPsi}}^2) \leq \lambda_{\max}\abs{\braket{\Psi|\tPsi}}^2 + \teps(1 - \abs{\braket{\Psi|\tPsi}}^2), \\
\end{split}
\end{equation}
leading to $\abs{\braket{\Psi|\tPsi}}^2 \geq 1-2\teps$. Without loss of generality, assume the largest and the smallest eigenvalues of $H$ is $\frac{\Delta H}{2}$ and $-\frac{\Delta H}{2}$. Then 
\begin{equation}
\label{eq:pure-rld-2}
\begin{split}
\abs{V_{H}(\tPsi) - V_{H}(\Psi)} &\leq
\abs{\trace\big(H(\ket{\tPsi}\bra{\tPsi}-\ket{\Psi}\bra{\Psi})H\ket{\tPsi}\bra{\tPsi}\big)}\\
&\qquad + \abs{\trace\big((\ket{\tPsi}\bra{\tPsi}-\ket{\Psi}\bra{\Psi})(H^2-H\ket{\tPsi}\bra{\tPsi}H\big)}\\
&\leq \norm{\ket{\tPsi}\bra{\tPsi}-\ket{\Psi}\bra{\Psi}}_1 \left(\braket{\Psi|H^2|\Psi} + \norm{H^2-H\ket{\tPsi}\bra{\tPsi}H} \right)\\
&\leq 2\sqrt{2\teps} \cdot \frac{3(\Delta H)^2}{4}, 
\end{split}
\end{equation}
where in the second step we use $\abs{\trace(AB)}\leq \norm{A}_1\norm{B}$ and in the third step we use $\norm{\ket{\tPsi}\bra{\tPsi}-\ket{\Psi}\bra{\Psi}}_1 = 2\sqrt{1 - \abs{\braket{\Psi|\tPsi}}^2}\leq 2\sqrt{2\teps}$. \eqref{eq:pure-rld} then follows from \eqref{eq:pure-rld-1} and \eqref{eq:pure-rld-2}. 
\end{proof}

Given \lemmaref{lemma:coh} and \lemmaref{lemma:pure-rld}, we are now ready to present our resource-theoretic bound on the code infidelity. 
The intuition on the operational level is that the accuracy of the optimal recovery transformation, which can be understood as a resource conversion task, is limited by the monotonicity of RLD QFI that diverges as approaching the pure target state  (see \figbref{fig:coh}).
More specifically, we compare the resource monotone, i.e., the RLD QFI of the decoded coherent state with that of the original noisy coherent state. The lower bound then follows from the monotonicity of the RLD QFI and the fact that the RLD QFI of the decoded coherent state is lower bounded by $O(1/\epsilon)$. 

Note that we will also use the RLD QFI for quantum channels defined by~\cite{hayashi2011comparison,katariya2020geometric}
\begin{equation}
\label{eq:RLD-def}
\begin{split}
&F_\txr(\mN_\theta) = \begin{cases}
\norm{\trace_{S(\mN_\theta)}\big((\ptheta \Gamma^{\mN_\theta})(\Gamma^{\mN_\theta})^{-1}(\ptheta \Gamma^{\mN_\theta}) \big)} & \text{(R),}\\
+\infty & \text{otherwise},\\
\end{cases}\\
&\;\;\; \text{(R):~~} {\rm span}\{\ptheta K_{i,\theta},\forall i\} \subseteq {\rm span}\{K_{i,\theta},\forall i\},
\end{split}
\end{equation}
where $\mN_{\theta}(\cdot) = \sum_{i=1}^r K_{i,\theta}(\cdot) K_{i,\theta}^\dagger$. Here we use the Choi operator of $\mN_\theta$: $\Gamma^{\mN_\theta} = (\mN_\theta\otimes \id)(\Gamma)$, where $\Gamma  = \ket{\Gamma}\bra{\Gamma}$ and $\ket{\Gamma} = \sum_{i} \ket{i}\ket{i}$. $S(\mN_\theta)$ denotes the output system of $\mN_\theta$. $F_\txr(\mN_\theta)$ is also additive~\cite{hayashi2011comparison}. 
The result we obtained is the following. 

\begin{theorem}
\label{lemma:coh}
Consider a covariant code $\mE_{\LtoS}$ under a noise channel $\mN_S(\cdot) = \sum_{i=1}^r K_{S,i} (\cdot) K_{S,i}^\dagger$. If $\mN_S$ commutes with $\mU_{S,\theta}$, 
$U_{L,S}(\theta)$ share a common period and 
\begin{equation}
\label{eq:cond-2}
{\rm span}\{K_{S,i} H_S,\forall i\} \subseteq {\rm span}\{K_{S,i},\forall i\},
\end{equation}
then $\epsilon$ and $\epsilon_\Choi$ are lower bounded as follows, 
\begin{equation}
\label{eq:bound-3}
\epsilon \geq \ell_2\left(\frac{(\Delta H_L)^2}{4 F_\txr(\mN_S,H_S)}\right) = \frac{(\Delta H_L)^2}{4 F_\txr(\mN_S,H_S)} + O\left(\left(\frac{(\Delta H_L)^2}{4 F_\txr(\mN_S,H_S)}\right)^2\right), 
\end{equation}
and 
\begin{equation}
\label{eq:bound-3-avg}
\epsilon_\Choi \geq \ell_3\left(\frac{\trace(H_L^2)}{d_L F_\txr(\mN_S,H_S)}\right) = \frac{\trace(H_L^2)}{d_L F_\txr(\mN_S,H_S)} + O\left(\left(\frac{\trace(H_L^2)}{d_L F_\txr(\mN_S,H_S)}\right)^2\right), 
\end{equation}
where $\ell_2(x)= x + O(x^2)$ is the inverse function of the monotonic increasing function $x = 1/({(1 - 3\epsilon + \epsilon^2)(1 - 6\sqrt{2\epsilon})})$ for $\epsilon \in [0,1/72)$, $\ell_3(x)= x + O(x^2)$ is the inverse function of the monotonic increasing function $x = 1/({(1 - 3\epsilon + \epsilon^2)(1 - \frac{3d_L(\Delta H_L)^2}{2\trace(H_L^2)}\sqrt{2\epsilon})})$ for $\epsilon \in [0,\frac{2\trace(H_L^2)^2}{9d_L^2(\Delta H_L)^4})$, and $F_\txr(\mN_S,H_S)$ is the RLD QFI of $\mN_{S,\theta}$. 
Specifically, $F_\txr(\mN_S,H_S) = \big\|\trace_S\big(\Gamma_{SR}^{\mN_S,H_S}(\Gamma_{SR}^{\mN_S})^{-1}\Gamma_{SR}^{\mN_S,H_S}\big)\big\|$, with  $\Gamma_{SR}^{\mN_S} = (\mN_S \otimes \id_R)(\Gamma_{SR})$
and 
$\Gamma_{SR}^{\mN_S,H_S} = (\mN_S \otimes \id_R)\big((H_S\otimes \id_R)\Gamma_{SR}\big) - (\mN_S \otimes \id_R)\big(\Gamma_{SR}(H_S\otimes \id_R)\big)$, 
where $\Gamma_{SR} = \ket{\Gamma}_{SR}\bra{\Gamma}_{SR}$ and $\ket{\Gamma}_{SR} = \sum_i \ket{i}_S\ket{i}_R$. 
\end{theorem}
\begin{proof}
Let $\ket{+_L} = \frac{\ket{0_L} + \ket{1_L}}{\sqrt{2}}$. Then according to \lemmaref{thm:cov}, there exists a covariant recovery channel $\mR_{\StoL}^\cov$ such that 
\begin{equation}
1 - \braket{+_L|\rho_L|+_L} \leq \epsilon,
\end{equation}
where $\rho_L = ( \mR_{\StoL}^\cov \circ \mN_S \circ \mE_{\LtoS}) (\ket{+_L}\bra{+_L})$. 
According to \lemmaref{lemma:pure-rld}, 
\begin{equation}
\label{eq:variance}
F_\txr(\rho_L,H_L) \geq \frac{1 - 3\epsilon + \epsilon^2}{\epsilon} \cdot \left(V_{H_L}(\ket{+_L} - \frac{3\sqrt{2\epsilon}(\Delta H_L)^2}{2}\right) ,
\end{equation}
where the variance $V_{H_L}(\ket{+_L}) = \braket{+_L|H_L^2|+_L} - \braket{+_L|H_L|+_L}^2 = \frac{(\Delta H_L)^2}{4}$. $\epsilon < 1/72$ guarantees the right-hand side is positive. On the other hand, using \eqref{eq:monotone-RLD},  
\begin{equation}
\label{eq:channel-RLD}
F_\txr(\rho_L,H_L) 
\leq F_\txr(\rho_S,H_S)
= F_\txr(\mN_{S,\theta}(\mE_{\LtoS}(\ket{+_L}\bra{+_L})) ) \leq  F_\txr(\mN_{S,\theta}),
\end{equation}
where $\rho_S = (\mN_S \circ \mE_{\LtoS}) (\ket{+_L}\bra{+_L})$ and 
\begin{equation}
F_\txr(\mN_{S,\theta}) = 
\begin{cases}
F_\txr(\mN_{S},H_S) & \text{\eqref{eq:cond-2}},\\
+\infty & \text{otherwise}. \\
\end{cases}
\end{equation} 
Using \eqref{eq:variance} and \eqref{eq:channel-RLD}, we have
\begin{equation}
\epsilon \cdot \frac{1}{(1 - 3\epsilon + \epsilon^2)(1 - 6\sqrt{2\epsilon})} \geq \frac{(\Delta H_L)^2}{4 F_\txr(\mN_S,H_S)}.  
\end{equation}
\eqref{eq:bound-3} then follows from the fact that $\ell_2(\cdot)$ is the inverse function of $x = \frac{1}{(1 - 3\epsilon + \epsilon^2)(1 - 6\sqrt{2\epsilon})}$. 

Similarly, let $\ket{\gamma_{LR}} = \frac{1}{\sqrt{d_L}}\sum_{i=1}^{d_L} \ket{i}_L \ket{i}_R$. Then according to \lemmaref{thm:cov}, there exists a covariant recovery channel $\mR_{\StoL}^\Choi$ such that 
\begin{equation}
1 - \braket{\gamma_{LR}|\rho_{LR}|\gamma_{LR}} \leq \epsilon_\Choi,
\end{equation}
According to \lemmaref{lemma:pure-rld}, 
\begin{equation}
\begin{split}
F_\txr(\rho_{LR},H_L\otimes\id_R) 
&\geq \frac{1 - 3\epsilon + \epsilon^2}{\epsilon} \cdot \left( V_{H_L\otimes\id_R}(\ket{\gamma_{LR}}) - \frac{3\sqrt{2\epsilon}(\Delta H_L)^2}{2}\right) \\
&\geq \frac{1 - 3\epsilon + \epsilon^2}{\epsilon} \cdot \left( \frac{\trace(H_L^2)}{d_L} - \frac{3\sqrt{2\epsilon}(\Delta H_L)^2}{2} \right). 
\end{split}
\end{equation}
The rest of the proof is exactly the same as in the proof of the lower bound on the worst-case infidelity. 
\end{proof}

As mentioned before, it is always true that $F_{\txr}(\mN_{S,\theta}) \geq F_{\txs}^{\reg}(\mN_{S,\theta})$. Moreoever, $\ell_3(x) \leq \ell_2(x) \leq \ell_1(x)$ for all $x$,\footnote{Note that when $d_L = 2$, $\ell_2(x) = \ell_3(x)$.} 
therefore \thmref{thm:SLD} provides a tighter bound on the code infidelity than \thmref{lemma:coh}. 
Note that the resource theory approach operates on the state level, which is fundamentally different from the metrological approach.   The assumptions and quantities involved in the two approaches are also different.
From the proof of \thmref{lemma:coh}, we see that there are three advantages of the resource-theoretic approach: (1) We can also derive a lower bound on the Choi code infidelity. (2) We can replace the entanglement-assisted RLD QFI with the one without entanglement assistance $\max_\rho F_\txr(\mN_\theta(\rho))$ in \eqref{eq:bound-3} which might tighten it in certain scenarios. (3) The lower bounds are state-dependent, e.g., we are allowed to replace $F_\txr(\mN_{S},H_S)$ with $F_\txr((\mN_S \circ \mE_{\LtoS}) (\ket{+_L}\!\bra{+_L}),H_S)$ in \eqref{eq:bound-3}, which may be of independent interest in determining the lower bounds on the code infidelity for some special types of covariant codes. 


\section{Local Hamiltonians and local noises}

One of the most common scenarios where covariant codes are considered is when $S$ is an $n$-partite system, consisting of subsystems $S_1,S_2,\ldots,S_n$. The physical Hamiltonian and the noise channel are both local, given by 
\begin{gather}
H_S = \sum_{k=1}^n H_{S_k},\quad
\mN_S = \bigotimes_{k=1}^n \mN_{S_k},\quad 
\mN_{S_k}(\cdot) = \sum_{i=1}^{r_k} K_{S_k,i} (\cdot) K_{S_k,i}^\dagger. 
\end{gather}

In general, it takes exponential time (in the number of subsystems) to solve our lower bounds on the code infidelity. However, when the Hamiltonians and the noises are local, using the additivity of channel QFIs, we could directly calculate the lower bounds, requiring only computation of the channel QFIs in each subsystem. To be specific, for $\epsilon$-correctable codes under $\mN_{S}$, \thmref{thm:SLD} indicates that when 
\begin{equation}
\label{eq:cond-1-local}
H_{S_k} \in {\rm span}\{K_{S_k,i}^\dagger K_{S_k,j},\forall i,j\},\quad \forall k, 
\end{equation}
we have
\begin{equation}
\label{eq:bound-1-local}
\epsilon \geq \ell_1\left(\frac{(\Delta H_L)^2}{4 \sum_{k=1}^n F_\txs^\reg(\mN_{S_k},H_{S_k})}\right).
\end{equation}

Instead of finding bounds for local noise channels $\mN_S$ with certain noise rates, we sometimes are more interested in the capability of a code to correct single errors (each described by $\mM_{S_k}$). Consider the single-error noise channel 
\begin{equation}
\mM_S = \sum_{k=1}^n q_k \mM_{S_k}, \quad \sum_{k=1}^n q_k = 1,
\end{equation}
where $q_k$ is the probability that an error $\mM_{S_k}$ occurs on the $k$-th subsystem. In order to obtain lower bounds on the code infidelity under noise channels $\mM_S$, we use the following local noise channel 
\begin{equation}
\mN_S(\delta) = \bigotimes_{k=1}^n \mN_{S_k}(\delta) = \bigotimes_{k=1}^n \big((1-\delta q_k)\id + \delta q_k \mM_{S_k}\big)
= (1-\delta)\id + \delta \sum_{k=1}^n q_k \mM_{S_k} + O(\delta^2),
\end{equation}
whose local noise rates are proportional to a small positive parameter $\delta$. Using the concavity of $f^2(\Phi,\id)$ with respect to the channel $\Phi$, we have 
\begin{equation}
f^2(\mR_{\StoL} \circ \mN_{S}(\delta) \circ  \mE_{\LtoS},\id_L) 
\geq (1-\delta) + \delta f^2(\mR_{\StoL} \circ \mM_{S} \circ  \mE_{\LtoS},\id_L) + O(\delta^2).
\end{equation}
Taking the limit $\delta \rightarrow 0^+$, we must have $\epsilon(\mM_S,\mE_{\LtoS}) \geq \liminf_{\delta \rightarrow 0^+} \frac{1}{\delta}\cdot{\epsilon(\mN_S(\delta),\mE_{\LtoS})}$. Therefore, for $\epsilon$-correctable codes under single-error noise channels $\mM_{S}$, \thmref{thm:SLD} indicates that when \eqref{eq:cond-1-local} is satisfied,
\begin{equation}
\label{eq:single-1}
\epsilon \geq 
\liminf_{\delta \rightarrow 0^+} \,\ell_1\left(\frac{(\Delta H_L)^2}{4 \delta \sum_{k=1}^n F_\txs^\reg(\mN_{S_k}(\delta),H_{S_k})}\right).
\end{equation} 
Note that discussions here analogously apply to 
\thmref{lemma:coh} due to the additivity of the channel RLD QFI, although we will only focus on \thmref{thm:SLD} starting now since it provides the tightest bound in the following scenarios.

For generic error models with noises on multiple subsystems, e.g., the random phase error model in \cite{woods2019continuous}, the regularized SLD QFI and thus the lower bound on the code infidelity could be difficult to compute. In \appref{app:multi-error}, we derive an upper bound on the regularized SLD QFI for multi-error noise channels which is efficiently computable in certain scenarios and provide an example of multiple erasure errors and local Hamiltonians.

\subsection{Erasure noise}

Now we present our bounds for the local erasure noise channel $\mN^{\mathrm e}(\rho) = (1-p)\rho + p \ket{\vac}\bra{\vac}$ on each subsystem. Here $p$ is the noise rate and we use the vacuum state $\ket{\vac}$ to represent the state of the erased subsystems. The Kraus operators for $\mN^{\mathrm e}$ are
\begin{equation}
\label{eq:Kraus-erasure}
K_{1} = \sqrt{1-p}\id, \quad K_{i+1} = \sqrt{p}\ket{\vac}\bra{i}, \quad \forall 1 \leq i \leq d. 
\end{equation}
Different subsystems can have different noise rates $p_k$ and dimensions $d_k$. 
As derived in \appref{app:QFI}, the regularized SLD QFI for erasure noise is 
\begin{equation}
\label{eq:QFI-erasure}
F_\txs^\reg(\mN^{\mathrm e},H) = (\Delta H)^2  \frac{1-p}{p}. 
\end{equation}
For $\epsilon$-correctable codes under local erasure noise channel $\mN^{\mathrm e}_S = \bigotimes_{k=1}^n \mN^{\mathrm e}_{S_k}$, we have 
\begin{equation}
\label{eq:eb-1}
\epsilon 
\geq \ell_1\left(\frac{(\Delta H_L)^2}{4\sum_{k=1}^n \frac{1-p_k}{p_k}(\Delta H_{S_k})^2}\right),
\end{equation}
using \eqref{eq:bound-1-local}. For $\epsilon$-correctable codes under single-error erasure noise channel $\mM^{\mathrm e}_S = \sum_{k=1}^n q_k \mM^{\mathrm e}_{S_k}$ where $\mM^{\mathrm e}_{S_k}(\rho_{S_k}) = \ket{\vac}\bra{\vac}_{S_k}$,
\begin{equation}
\label{eq:eb-2}
\epsilon 
\geq \ell_1\left(\frac{(\Delta H_L)^2}{4\sum_{k=1}^n \frac{1}{q_k}(\Delta H_{S_k})^2}\right),
\end{equation}
using \eqref{eq:single-1}. In particular, when the probability of erasure is uniform on each subsystem, i.e., $q_k = \frac{1}{n}$, we have 
\begin{equation}
\label{eq:eb-3}
\epsilon 
\geq \ell_1\left(\frac{(\Delta H_L)^2}{4n\sum_{k=1}^n(\Delta H_{S_k})^2}\right) = \frac{(\Delta H_L)^2}{4n\sum_{k=1}^n(\Delta H_{S_k})^2} + O\left(\left(\frac{(\Delta H_L)^2}{4n\sum_{k=1}^n(\Delta H_{S_k})^2}\right)^2\right).
\end{equation}
As a comparison, Theorem 1 in~\cite{faist2019continuous} says that 
\begin{equation}
\epsilon \geq \frac{(\Delta H_L)^2}{4n^2  \max_k (\Delta H_{S_k})^2}.
\end{equation}
Our bound \eqref{eq:eb-3} has a clear advantage in the small infidelity limit by improving the maximum of $\Delta H_{S_k}$ to their quadratic mean. A direct implication of \eqref{eq:eb-3} is an improved approximate Eastin--Knill theorem which establishes the infidelity lower bound for covariant codes with respect to special unitary groups (see \appref{app:ek}).

\subsection{Depolarizing noise}

Next, we present our bounds for local depolarizing noise channel $\mN^{\mathrm d}(\rho) = (1-p)\rho + p \frac{\id}{d}$ on each subsystem, which has not been studied before. Again, we assume different subsystems can have different noise rates $p_k$ and dimensions $d_k$. The Kraus operators for $\mN^{\mathrm d}$ are 
\begin{equation}
K_{1} = \sqrt{1 - \frac{d^2 - 1}{d^2} p}\id,\quad K_{i} = \sqrt{\frac{p}{d^2}} U_{i-1},\quad \forall 2 \leq i \leq d^2, 
\end{equation}
where 
$\{U_0 = \id,U_1,\ldots,U_{d_k^2 - 1}\}$ is a unitary orthonormal basis in $\bC^{d \times d}$. 

In order to apply \thmref{thm:SLD}, we need to solve the following SDP
\begin{equation}
F_\txs^\reg(\mN^{\mathrm d},H) = \min_{h:\beta = 0}4\norm{\alpha},
\end{equation}
where $\beta = \vK^\dagger h \vK - H$ and $\alpha = \vK^\dagger h^2 \vK - H^2$.

When $d = 2$, as shown in \appref{app:QFI}, we have 
$
F_\txs^\reg(\mN^{\mathrm d},H) =  (\Delta H)^2 \frac{2(1-p)^2}{p(3-2p)}. 
$
When all subsystems are qubits, for $\epsilon$-correctable codes under local depolarizing noise channels $\mN^{\mathrm d}_S = \bigotimes_{k=1}^n \mN^{\mathrm d}_{S_k}$,
\begin{equation}
\epsilon 
\geq \ell_1\left(\frac{(\Delta H_L)^2}{4 \sum_{k=1}^n \frac{2(1-p_{k})^2}{p_k(3-2p_k)} (\Delta H_{S_k})^2}\right),
\end{equation}
using \eqref{eq:bound-1-local} and for $\epsilon$-correctable codes under single-error depolarizing noise channels $\mM^{\mathrm d}_S = \sum_{k=1}^n q_k \mM^{\mathrm d}_{S_k}$ where $\mM^{\mathrm d}_{S_k}(\rho_{S_k}) = \frac{\id}{2}$, 
\begin{equation}
\label{eq:bound-depolarizing-1}
\epsilon 
\geq \ell_1\left(\frac{3 (\Delta H_L)^2}{8 \sum_{k=1}^n \frac{1}{q_k} (\Delta H_{S_k})^2}\right), 
\end{equation}
using \eqref{eq:single-1}. 

The situation is more complicated when $d > 2$, because the regularized SLD QFI may not have a closed-form expression. Instead, we show in \appref{app:SLD-depolarizing} that 
\begin{equation}
\label{eq:SLD-depolarizing}
F^\reg_\txs(\mN^{\mathrm d},H) \!\leq\! (\Delta H)^2 \frac{(1-p)^2}{p(1 \!+\! \frac{2}{d^2} \!-\! p)} \!\leq\! (\Delta H)^2 \frac{1-p}{p},
\end{equation}
by choosing a special $h$ which satisfies $\beta = 0$ to calculate an upper bound on $4\min_{h:\beta=0}\norm{\alpha}$. Note that the right-hand side of \eqref{eq:SLD-depolarizing} is equal to the regularized SLD QFI for erasure channels \eqref{eq:QFI-erasure}. We conclude that \eqsref{eq:eb-1}{eq:eb-3} hold true for general depolarizing channels as well, regardless of the dimensions of subsystems.
We also remark that the upper bound on the regularized SLD QFI for depolarizing channels we derived here might be of independent interest in quantum metrology. 

Note that the channel RLD QFI $F_\txr(\mN_{S_k}^{\mathrm d},H_{S_k})$ also upper bounds $F^\reg_\txs(\mN_{S_k}^{\mathrm d},H_{S_k})$ and has a closed-form expression.  
The RLD QFI $F_\txr(\mN^{\mathrm d},H)$ for depolarizing channels could be directly calculated using 
\begin{equation}
F_\txr(\mN^{\mathrm d},H) = \big\|\trace_{S(\mN^{\mathrm d})}\big(\Gamma^{\mN^{\mathrm d},H}(\Gamma^{\mN^{\mathrm d}})^{-1}\Gamma^{\mN^{\mathrm d},H}\big)\big\|,
\end{equation}
where $\Gamma^{\mN^{\mathrm d}} = (\mN^{\mathrm d}\otimes\id)\Gamma$, $\Gamma^{\mN^{\mathrm d},H} = (\mN^{\mathrm d}\otimes\id)(\ket{H}\bra{\Gamma} - \ket{\Gamma}\bra{H})$ and $\ket{H} = (H\otimes \id)\ket{\Gamma}$. 
Then 
\begin{gather}
F_\txr(\mN^{\mathrm d},H) = \frac{(1-p)^2}{4(1 - \frac{d^2 - 1}{d^2} p)} (\Delta H)^2
+ \frac{d(1-p)^2}{p} \trace(H^2),
\end{gather}
where the second term increases linearly with respect to $d$, meaning the RLD QFI only provides a close bound for a small subsystem dimension. 

\section{Example: Thermodynamic codes}

Finally, we provide an example that saturates the lower bound for single-error erasure noise channels 
in the small infidelity limit 
and matches the scaling of the lower bound for single-error depolarizing noise channels,
while previously only the scaling optimality for erasure channels was demonstrated~\cite{faist2019continuous}.

We consider the following two-dimensional thermodynamic code~\cite{brandao2019quantum,faist2019continuous,ouyang2019robust}
\begin{gather}
\mE_{\LtoS}(\ket{0_L}) = \ket{g_0} = \ket{m_n}, 
\\
\mE_{\LtoS}(\ket{1_L}) = \ket{g_1} = \ket{(-m)_n}, 
\end{gather} 
where
\begin{equation}
\ket{(\pm m)_n} = \binom{n}{\frac{n\pm m}{2}}^{-\frac{1}{2}}\sum_{\vj:\sum_k j_k = \pm m} \ket{\vj}, 
\end{equation}
and $\vj = (j_1,j_2,\ldots,j_n) \in \{-1,1\}^n$. The logical subspace is spanned by two Dicke states with different values of the total angular momentum along the z-axis. We also assume $n+m$ is an even number and $3 \leq m \ll N$. 
It is a covariant code whose physical and logical Hamiltonians are
\begin{equation}
H_S = \sum_{k=1}^n (\sigma_z)_{S_k},\quad 
H_L = m Z_L,
\end{equation}
where $\sigma_z = \ket{1}\bra{1} - \ket{-1}\bra{-1}$. 

Let $\ket{g^{(k)}_{0,\pm 1}} = \ket{(m\pm 1)_{n-1}}_{S\backslash S_k}\ket{\vac}_{S_k}$, $\ket{g^{(k)}_{1,\pm 1}} = \ket{(- m\pm 1)_{n-1}}_{S\backslash S_k}\ket{\vac}_{S_k}$, which represent the logical states after an erasure error occurs on $S_k$, and $\Pi^\perp$ be the projector onto the orthogonal subspace of ${\rm span}\{\ket{g^{(k)}_{0,\pm 1}},\ket{g^{(k)}_{1,\pm 1}},\forall k\}$. Consider the erasure noise channel $\mM_S = \frac{1}{n}\sum_{k=1}^n \mM_{S_k}$ where $\mM_{S_k}(\rho_{S_k}) = \ket{\vac}\bra{\vac}_{S_k} $ and the recovery channel
\begin{equation}
\mR_{\StoL}(\rho_S) = \sum_{k=1}^n \sum_{i,i'=0}^1 \sum_{j=\pm 1}  \ket{g_i}\bra{g^{(k)}_{i,j}}\rho_S\ket{g^{(k)}_{i',j}}\bra{g_{i'}} + \trace(\Pi^\perp \rho_S \Pi^\perp) \ket{g_0}\bra{g_0},  
\end{equation}
which maps the state $\ket{g^{(k)}_{i,\pm 1}}$ to $\ket{g_{i}}$ for all $k$. 
Then we could verify that $\mR_{\StoL} \circ \mM_{S} \circ \mE_{\LtoS} = \mD_{p,0}$ with $p = \frac{1}{2}\big( 1 - \sqrt{1 - \frac{m^2}{n^2}}\big)$. 
Using the relation between the noise rate $p$ and the worst-case entanglement fidelity of a dephasing channel (see \appref{app:dephasing-fidelity}), we must have 
\begin{align}
\epsilon(\mM_S,\mE_{\LtoS}) &\leq 1 - f^2(\mR_{\StoL} \circ \mM_{S} \circ \mE_{\LtoS},\id_L) \\ &= \frac{1}{2}\bigg( 1 - \sqrt{1 - \frac{m^2}{n^2}}\bigg) 
 = \frac{m^2}{4n^2} + O\bigg(\frac{m^4}{n^4}\bigg).
\end{align}
On the other hand, the lower bound (\eqref{eq:eb-3}) for $\epsilon = \epsilon(\mM_S,\mE_{\LtoS})$ is given by 
\begin{equation}
\epsilon 
\geq \ell_1\left(\frac{m^2}{4n^2}\right) = \frac{m^2}{4n^2} + O\left(\frac{m^4}{n^4}\right),
\end{equation}
which is saturated asymptotically when $m/N \rightarrow 0$. 

Next, we consider the single-error depolarizing noise channel $\mM_S = \frac{1}{n}\sum_{k=1}^n \mM_{S_k}$ where $\mM_{S_k}(\rho_{S_k}) = \frac{\id}{2}$. It is in general difficult to write down the optimal recovery map explicitly. Instead, in order to calculate $\epsilon(\mM_S,\mE_{\LtoS})$, we apply Corollary 2 in \cite{beny2010general} to calculate an upper bound on the infidelity of thermodynamic codes in the limit $m/N\rightarrow 0$ and we obtain (see details in \appref{app:infidelity})
\begin{equation}
\epsilon(\mM_S,\mE_{\LtoS}) \leq \frac{3m^2}{4n^2} + o\bigg(\frac{m^2}{n^2}\bigg), 
\end{equation}
which also matches the scaling of our lower bound for depolarizing noise channels (\eqref{eq:bound-depolarizing-1}), i.e., $\epsilon(\mM_S,\mE_{\LtoS}) \geq \frac{3m^2}{8n^2} + O\big(\frac{m^4}{n^4}\big)$.

\section{Conclusions and outlook}

In this paper, we advanced the understanding of covariant QEC by leveraging insights and techniques from quantum metrology and quantum resource theory. We first presented covariant QEC as a special type of metrological protocol where the sensitivity in parameter estimation could be linked to the code infidelity. We took inspirations from recent developments in quantum channel estimation: a no-go theorem~\cite{escher2011general,demkowicz2012elusive,demkowicz2014using,yuan2017quantum,demkowicz2017adaptive,zhou2018achieving,zhou2020theory} on the existence of perfect QEC was discovered based on a relation between Hamiltonians and noises (the HKS condition) which leads to a no-go theorem for covariant QEC; efficiently computable QFIs of quantum channels were also proposed~\cite{zhou2020theory,hayashi2011comparison,katariya2020geometric}, which leads to efficiently computable lower bounds for the code infidelity under generic noise channels. We also demonstrated how covariant QEC 
can be understood from an operational resource theory perspective, where the key insight is that there are fundamental limits on the distillation of pure coherent states using noisy ones~\cite{marvian2020coherence,FangLiu19:nogo}. 
The lower bounds we derived not only have a broad range of applications, but also improve upon previous lower bounds, which also lead to an improved approximate Eastin-Knill theorem that may be of particular interest in quantum computation.

In our metrological proof of the infidelity lower bounds, we reduced noisy quantum channels to rotated dephasing channels using one noiseless ancillary qubit. It in turn provided an entanglement-assisted metrological protocol for channel estimation which might be of independent interest in quantum metrology. One implication of it is that known covariant codes might help improve the lower bounds for the channel QFIs, in situations where they are hard to calculate. Conversely, it indicates that lower bounds on the code infidelity might be improved if a separation between the entanglement-assisted QFIs with respect to one noiseless ancillary qubit and those with respect to an unbounded ancillary system could be identified.

There are still many open questions and future directions in the study of covariant QEC. First, it is not known, whether the HKS condition, which was shown to be sufficient for the non-existence of perfect covariant QEC codes, is also necessary. There are some examples of perfect covariant QEC codes, such as the [[4,2,2]] QEC code under single-qubit erasure noise~\cite{gottesman2016quantum,faist2019continuous}, repetition codes under bit-flip noise~\cite{kessler2014quantum,arrad2014increasing,dur2014improved,reiter2017dissipative}, but it is not yet clear how to generalize those examples. On the other hand, when the HKS condition is satisfied, it would also be desirable to obtain a systematic procedure to construct covariant codes saturating the infidelity lower bounds, at least in terms of scaling~\cite{woods2019continuous}.  From the resource theory perspective, it would be interesting to investigate whether different monotones may induce other useful bounds, and whether invoking the resource theory of channels (see e.g., \cite{LiuWinter19,LiuYuan19}) approaches may lead to new insights.  It would also be important to further explore possible implications of the limitations on covariant QEC for physics, where symmetries naturally play prominent roles in a wide range of scenarios.

\bigskip
\noindent
\textbf{Note added.}  During the completion of this work, an independent work by Kubica and Demkowicz-Dobrzanski~\cite{kubica2020using} appeared on arXiv, where a lower bound on the infidelity of covariant codes was also derived using tools from quantum metrology. Note that we employed different techniques and obtained lower bounds with a quadratic advantage in the small infidelity limit over the one in~\cite{kubica2020using}. The bound in~\cite{kubica2020using} was recently improved in~\cite{yang2020covariant} by quantifying the code inaccuracy using the diamond norm. Note that our bound still performs better in the small infidelity limit because $1 - f^2(\ket{\psi},\rho) \leq \frac{1}{2}\norm{\ket{\psi}\bra{\psi} - \rho}_1$.

\section*{Acknowledgments}

We thank Victor V. Albert, Sepehr Nezami, John Preskill, Beni Yoshida, David Layden and Junyu Liu for helpful discussions. SZ and LJ acknowledge support from the ARL-CDQI (W911NF-15-2-0067), ARO (W911NF-18-1-0020, W911NF-18-1-0212), ARO MURI (W911NF-16-1-0349), AFOSR MURI (FA9550-15-1-0015, FA9550-19-1-0399), DOE (DE-SC0019406), NSF (EFMA-1640959, OMA-1936118), and the Packard Foundation (2013-39273).  ZWL is supported by Perimeter Institute for Theoretical Physics.
Research at Perimeter Institute is supported by the Government
of Canada through Industry Canada and by the Province of Ontario through the Ministry of Research and Innovation.

\appendix

\renewcommand{\thesection}{\Alph{section}}
\setcounter{theorem}{0}
\renewcommand{\thetheorem}{A\arabic{theorem}}

\appendix

\section{Additivity of the regularized SLD QFI}
\label{app:additivity}

Here we prove the additivity of the regularized SLD QFI:
\begin{equation}
F_\txs^{\reg}(\mN_{\theta}\otimes \tilde\mN_{\theta}) = F_\txs^{\reg}(\mN_{\theta}) + F_\txs^{\reg}(\tilde\mN_{\theta}),
\end{equation}
for arbitrary quantum channels $\mN_{\theta}$ and $\tilde\mN_{\theta}$. 

First, according to the additivity of the state QFI, we must have 
\begin{equation}
F_\txs^{\reg}(\mN_{\theta}\otimes \tilde\mN_{\theta}) \geq F_\txs^{\reg}(\mN_{\theta}) + F_\txs^{\reg}(\tilde\mN_{\theta}).
\end{equation}
Thus, we only need to prove 
\begin{equation}
F_\txs^{\reg}(\mN_{\theta}\otimes \tilde\mN_{\theta}) \leq F_\txs^{\reg}(\mN_{\theta}) + F_\txs^{\reg}(\tilde\mN_{\theta}).
\end{equation}
We use the following definition of the regularized SLD QFI~\cite{demkowicz2012elusive,demkowicz2014using,zhou2020theory} (which is equivalent to \eqref{eq:SLD-def})
\begin{equation}
\label{eq:SLD-def-1}
F_\txs^{\reg}(\mN_\theta) = 
\begin{cases}
4 \min_{\vK':\beta = 0} \norm{\alpha}, & i\sum_{i=1}^r (\ptheta K_i)^\dagger K_i \in {\rm span}\{K_i^\dagger K_j,\forall i,j\},\\
+\infty & \text{otherwise}, \\
\end{cases}
\end{equation}
where $\vK'$ is any set of Kraus operators representing $\mN_{\theta}$, $\alpha = \sum_{i=1}^r (\ptheta K'_i)^\dagger (\ptheta K'_i)$ and $\beta = i\sum_{i=1}^r (\ptheta K'_i)^\dagger K'_i$. 
Without loss of generality, assume both $F_\txs^{\reg}(\mN_{\theta})$ and $F_\txs^{\reg}(\tilde\mN_{\theta})$ are finite, i.e., $i\sum_{i=1}^r (\ptheta K_i)^\dagger K_i \in {\rm span}\{K_i^\dagger K_j,\forall i,j\}$ and $i\sum_{i=1}^{\tilde r} (\ptheta \tilde K_i)^\dagger \tilde K_i \in {\rm span}\{\tilde K_i^\dagger \tilde K_j,\forall i,j\}$. We first note that $F_\txs^{\reg}(\mN_{\theta}\otimes \tilde\mN_{\theta})$ is also finite, because
\begin{align}
i\sum_{i=1}^{r}\sum_{j=1}^{\tilde r} (\ptheta (K_i \otimes \tilde K_j))^\dagger (K_i \otimes \tilde K_j) 
&= i\sum_{i=1}^{r}(\ptheta K_i)^\dagger K_i \otimes \id
+ i\sum_{j=1}^{\tilde r} \id \otimes (\ptheta \tilde K_j)^\dagger \tilde K_j\\
&\in {\rm span}\{\id \otimes K_i^\dagger K_j,\tilde K_i^\dagger \tilde K_j \otimes \id,\forall i,j\}.
\end{align}
According to \eqref{eq:SLD-def-1}, there exists $\vK'$ and $\tilde \vK'$ such that $\beta = \tilde \beta = 0$ and 
\begin{equation}
F_\txs^\reg(\mN_\theta) = 4 \norm{\alpha},\quad 
F_\txs^\reg(\tilde \mN_\theta) = 4 \norm{\tilde \alpha}.
\end{equation}
Then $\tilde{\tilde{K}}'_{ij} = K_i' \otimes \tilde K_j'$ is a set of Kraus operators representing $\mN_\theta \otimes \tilde \mN_\theta$. 
\begin{gather}
\tilde{\tilde{\alpha}} = \sum_{i=1}^{r}\sum_{j=1}^{\tilde r} \ptheta(\tilde{\tilde{K}}_{ij})^\dagger  \ptheta(\tilde{\tilde{K}}_{ij})
= \alpha \otimes \id + \id \otimes \tilde \alpha + 2 \beta \otimes \tilde \beta = \alpha \otimes \id + \id \otimes \tilde \alpha, \\
\tilde{\tilde{\beta}} = \beta \otimes \id + \id \otimes \tilde{\beta} = 0. 
\end{gather}
Therefore $F^\reg_\txs(\mN_\theta\otimes\tilde \mN_\theta) \leq 4 \norm{\tilde{\tilde \alpha}} = 4 \norm{\alpha} + 4 \norm{\tilde \alpha} =  F_\txs^{\reg}(\mN_{\theta}) + F_\txs^{\reg}(\tilde\mN_{\theta})$.

\section{Worst-case entanglement fidelity for rotated dephasing channels}
\label{app:dephasing-fidelity}

Here we calculate the worst-case entanglement fidelity for rotated dephasing channels (\eqref{eq:dephasing}) 
\begin{equation}
\mD_{p,\phi}(\rho) = (1-p) e^{-i\frac{\phi}{2}Z} \rho e^{i\frac{\phi}{2}Z} + p e^{-i\frac{\phi}{2}Z} Z \rho Z e^{i\frac{\phi}{2}Z}.
\end{equation}
We use the following formula for the worst-case entanglement fidelity~\cite{schumacher1996sending}:
\begin{equation}
f^2(\mD_{p,\phi},\id) = \min_{\ket{\psi}} \bra{\psi}(\mD_{p,\phi}\otimes \id)(\ket{\psi}\bra{\psi})\ket{\psi}.
\end{equation}
Let $\ket{\psi} = \alpha_{00}\ket{00} + \alpha_{01}\ket{01} + \alpha_{10}\ket{10} + \alpha_{11}\ket{11}$, then 
\begin{multline}
(\mD_{p,\phi}\otimes \id)(\ket{\psi}\bra{\psi})
= \\
\begin{pmatrix}
\alpha_{00}\alpha_{00}^* & \alpha_{00}\alpha_{01}^* & (1-2p)e^{-i\phi}\alpha_{00}\alpha_{10}^* & (1-2p)e^{-i\phi} \alpha_{00}\alpha_{11}^* \\
\alpha_{00}\alpha_{01}^* & \alpha_{01}\alpha_{01}^* & (1-2p)e^{-i\phi}\alpha_{01}\alpha_{10}^* & (1-2p)e^{-i\phi} \alpha_{01}\alpha_{11}^* \\
(1-2p)e^{i\phi}\alpha_{10}\alpha_{00}^* & (1-2p)e^{i\phi}\alpha_{10}\alpha_{01}^* & \alpha_{10}\alpha_{10}^* & \alpha_{10}\alpha_{11}^* \\
(1-2p)e^{i\phi}\alpha_{11}\alpha_{00}^* & (1-2p)e^{i\phi}\alpha_{11}\alpha_{01}^* & \alpha_{11}\alpha_{10}^* & \alpha_{11}\alpha_{11}^* \\
\end{pmatrix}. 
\end{multline}
Then 
\begin{equation}
\begin{split}
1 - f^2(\mD_{p,\phi},\id) 
&= \max_{\alpha_{00,01,10,11}} 2\Re[(1 - (1-2p)e^{-i\phi})] (\abs{\alpha_{00}}^2+\abs{\alpha_{01}}^2)(\abs{\alpha_{10}}^2+\abs{\alpha_{11}}^2) \\
&= \frac{1}{2}(1 - (1-2p)\cos\phi) \geq p. 
\end{split}
\end{equation}

\section{An upper bound on the regularized SLD QFI for multi-error noise channels}
\label{app:multi-error}

Consider the following types of noise channels and physical Hamiltonians: 
\begin{equation}
\mM_S = \sum_{\chi \in \scrX} q_{\chi} \mM_\chi,
\quad 
H_S = \sum_{\chi \in \scrX} H_\chi. 
\end{equation}
where $q_\chi > 0$, $\sum_{\chi \in \scrX}^n q_\chi = 1$, $\scrX$ is a collection of sets of subsystems and $\mM_\chi$ and $H_\chi$ act on the corresponding sets of subsystems. For example, in the single-error case, $\scrX$ is the collection of all local subsystems $\{\{S_1\},\{S_2\},\ldots,\{S_n\}\}$.  

Assume that the HKS condition is satisfied for each $\chi$, i.e., $H_\chi \in {\rm span}\{K_{\chi,i}^\dagger K_{\chi,j},\forall i,j\}$ and that the set of Hamiltonians $\{H_\chi\}_{\chi \in \scrX}$ commute pairwise, i.e., $[H_\chi,H_{\chi'}] = 0$ for all $\chi,\chi' \in \scrX$. 
We then derive an upper bound of $F_\txs^\reg(\mM_S,H_S) = 4 \min_{h:\beta_S=0}\norm{\alpha_S}$, where 
\begin{equation}
\alpha_S 
= \vK_S^\dagger h^2\vK_S - H_S^2,
\text{~~and~~}
\beta_S = \vK_S^\dagger h \vK_S - H_S. 
\end{equation} 
When the HKS condition is satisfied for each $\chi$, to derive an upper bound on $F_\scrS^\reg(\mM_S,H_S)$, we can restrict $h$ to be a block diagonal matrix when partitioning the indices of Kraus operators according to $\scrX$. Then 
\begin{equation}
\alpha_S 
= \sum_\chi q_\chi \vK_\chi^\dagger h_\chi^2 \vK_\chi - \Big(\sum_\chi H_\chi\Big)^2,
\text{~~and~~}
\beta_S = \sum_\chi q_\chi \vK_\chi^\dagger h_\chi \vK_\chi - H_\chi. 
\end{equation} 
Let 
$\beta_\chi = q_\chi \vK_\chi^\dagger h_\chi \vK_\chi - H_\chi$. Then 
\begin{equation}
\label{eq:upper-multi}
\begin{split}
F_\txs^\reg(\mM_S,H_S) 
&= 4 \min_{h:\beta_S=0}\norm{\alpha_S} \leq 4 \min_{h_\chi:\beta_\chi=0}\norm{\alpha_S} \\&\leq 4 \min_{h_\chi:\beta_\chi=0}\norm{\sum_\chi q_\chi \vK_\chi^\dagger h_\chi^2 \vK_\chi}
\leq 4 \sum_{\chi \in \scrX} \min_{h_\chi:\beta_\chi=0} \norm{ q_\chi \vK_\chi^\dagger h_\chi^2 \vK_\chi} .
\end{split}
\end{equation}


The upper bound on the right-hand side of \eqref{eq:upper-multi} is efficiently computable when the size of each set $\chi$ is small. However, given $\scrX$ and $H_S$, it is in general not clear how to decompose $H_S$ into $\sum_{\chi \in \scrX} H_\chi$ in order to attain the optimal upper bound, and the upper bound might not be tight. 

Now we provide an example, where we compute the upper bound of the regularized SLD QFI for $t$ erasure errors and identical local Hamiltonians $H$. We have 
\begin{equation}
\mM_S = \sum_{\chi \in \scrX} q_\chi \mM_\chi, \quad H_S = \sum_{k=1}^n H_{S_k} = \sum_{\chi \in \scrX} H_\chi. 
\end{equation}
where $\scrX$ is the collection of all size-$t$ sets of subsystems, $\mM_\chi$ is the completely erasure channel on the corresponding set, $q_\chi = 1/\binom{n}{t}$ for all $\chi$, $H_\chi = \sum_{i=1}^t H_{\chi_i}/\binom{n-1}{t-1}$ and $H_{S_k}$ is identical for each subsystem which we denote by $H$. 
Suppose $\beta_\chi = q_\chi  \vK_\chi^\dagger h_\chi \vK_\chi - H_\chi$, where $K_{\chi,(i_1\cdots i_t)} = (\ket{\vac^{\otimes t}}_\chi\bra{i_1\cdots i_t}_\chi)$. Then $h_\chi = H_\chi/q_\chi$ and 
$q_\chi \vK_\chi^\dagger h_\chi^2 \vK_\chi = H_\chi^2/q_\chi$. 
We have 
\begin{equation}
\begin{split}
F_\txs^\reg(\mM_S,H_S) 
&\leq
4 \sum_{\chi \in \scrX} \min_{h_\chi:\beta_\chi=0} \norm{ q_\chi \vK_\chi^\dagger h_\chi^2 \vK_\chi} 
 \\ 
&= 4 \norm{H}^2  \binom{n}{t}^2\frac{t^2}{\binom{n-1}{t-1}^2} = 4\norm{H}^2 n^2 = (\Delta H)^2n^2, 
\end{split}
\end{equation}
It implies that
\begin{equation}
\epsilon \geq \ell_1\left(\frac{(\Delta H_L)^2}{4n^2(\Delta H)^2}\right),
\end{equation}
for arbitrary $t$, which matches our bound (\eqref{eq:eb-3}) in the $t = 1$ case. It is unclear though, whether the bound is tight. Note that a similar bound for $t$ erasure errors were derived in~\cite{yang2020covariant}.

\section{Regularized SLD QFI for erasure and single-qubit depolarizing channels}
\label{app:QFI}

Here we calculate the SLD QFI for erasure and depolarizing channels. We first calculate $F_\txs^\reg(\mN^{\mathrm e},H)$ where $\mN^{\mathrm e} = (1-p)\rho + p\ket{\vac}\bra{\vac}$. Using the Kraus operators in \eqref{eq:Kraus-erasure}, 
\begin{equation}
\beta = \vK^\dagger h \vK - H ~~\Leftrightarrow~~
h = 
\begin{pmatrix}
\frac{h_{11}}{1-p} & 0 \\
0 & \frac{H - h_{11}\id}{p}\\
\end{pmatrix}. 
\end{equation}
Then 
\begin{gather}
\alpha = \vK^\dagger h^2 \vK - H^2 =  \frac{h_{11}^2}{1-p} + \frac{(H - h_{11}\id)^2}{p} - H^2 = \frac{1-p}{p}H^2 - \frac{2h_{11}}{p} H + \frac{h_{11}^2}{p(1-p)},\\
\begin{split}
F_{\txs}^\reg(\mN^{\mathrm e},H) 
&= 4 \min_{h_{11}}\norm{\alpha}
= 4 \max_{\rho}\min_{h_{11}}\trace(\rho\alpha) \\
&= 4 \max_{\rho} \frac{1-p}{p} \big( \trace(H^2\rho) - \trace(\rho H)^2 \big) = \frac{1-p}{p} (\Delta H)^2,  
\end{split}
\end{gather}
where we use the minimax theorem~\cite{komiya1988elementary,do2001introduction} in the second step. 

We use the formula in Sec.~VII(A) in~\cite{zhou2020theory} to calculate the regularized SLD QFI $F_\txs^\reg(\mN^{\mathrm d},H)$ for single-qubit depolarizing channels $\mN^{\mathrm d}(\rho) = (1-p) \rho + p \frac{\id}{2}$. 
\begin{equation}
F^\reg_\txs(\mN^{\mathrm d},H) = (\Delta H)^2 \frac{1-w}{w},
\end{equation}
where $w = 4 \left(\frac{y^2}{2y} + \frac{xy}{x+y}\right)$ with $x = 1 - \frac{3}{4}p$ and $y = \frac{p}{4}$. Then $
F^\reg_\txs(\mN^{\mathrm d},H) = (\Delta H)^2 \frac{2(1-p)^2}{p(3-2p)}$.

\section{Regularized SLD QFI for general depolarizing channels}
\label{app:SLD-depolarizing}

Here we prove an upper bound on $F_\txs^\reg(\mN^{\mathrm d},H)$ for general depolarizing channels $\mN^{\mathrm d}(\rho) = (1-p) \rho + p \frac{\id}{d}$ with the Kraus operators
\begin{equation}
K_1 = \sqrt{x}\id,\quad K_i=\sqrt{y}U_{i-1},\forall 2 \leq i \leq d^2, 
\end{equation}
where we define $x = 1 - \frac{d^2 - 1}{d^2} p$, $y = \frac{1}{d^2}p$. 

Any $\tilde h$ satisfying $\tilde \beta = \vK^\dagger \tilde h \vK - H = 0$ provides an upper bound on $F_\txs^\reg(\mN^{\mathrm d},H)$ through
\begin{equation}
F_\txs^\reg(\mN^{\mathrm d},H) = 4 \min_{h:\beta = 0}\norm{\alpha} \leq 4 \norm{\alpha}|_{h=\tilde h}.
\end{equation}
To find a suitable $\tilde h$ which provides a good upper bound on $F_\txs^\reg(\mN^{\mathrm d},H)$, we use $\tilde h$ which is the solution of 
\begin{equation}
\label{eq:sdp}
4 \min_{h:\beta = 0} \trace(\alpha). 
\end{equation}
The solution of \eqref{eq:sdp} is 
\begin{equation}
\tilde h =\frac{1}{2 z d} 
\begin{pmatrix}
0 & \frac{\sqrt{xy}}{x+y} \trace(H U_1^\dagger U_0) & \cdots &\frac{\sqrt{xy}}{x+y} \trace(H U_{d^2-1}^\dagger U_0) \\
\frac{\sqrt{xy}}{x+y} \trace(H U_0^\dagger U_1) & 0 & \cdots & \frac{1}{2} \trace(H U_{d^2-1}^\dagger U_1)\\
\vdots & \vdots & \ddots & \vdots \\
\frac{\sqrt{xy}}{x+y} \trace(H U_0^\dagger U_{d^2-1}) & \frac{1}{2} \trace(H U_1^\dagger U_{d^2-1}) & \cdots & 0\\ 
\end{pmatrix},
\end{equation}
where $z = \frac{x y}{x+y} + \frac{y(d^2-2)}{4}$ and we used the assumption $\trace(H) = 0$ and 
\begin{equation}
\vK^\dagger \tilde{h}^2 \vK = \bigg( \frac{1}{4 z} - \frac{y}{4z^2}\bigg(\frac{1}{4} - \frac{xy}{(x+y)^2}\bigg) - 1\bigg) H^2 + \frac{y}{4 z^2 d} \bigg( \frac{x}{x+y} - \frac{1}{2}  \bigg)^2 \trace(H^2) \id. 
\end{equation}
Using $\|H^2\| = \frac{(\Delta H)^2}{4}$ and $\trace(H^2) \leq \frac{d}{4} (\Delta H)^2$,
\begin{equation}
F_\txs^\reg(\mN^{\mathrm d},H) \leq 4\norm{\alpha} \leq (\Delta H)^2 \bigg( \frac{1}{4z} - 1 \bigg) = (\Delta H)^2 \frac{d^2(1-p)^2}{p(d^2 (1-p) + 2)} \leq (\Delta H)^2 \bigg( \frac{1 - p}{p} \bigg),
\end{equation}
upper bounded by the $F_\txs^\reg(\mN^{\mathrm d},H)$ for erasure channels (\eqref{eq:QFI-erasure}).

\section{Improved approximate Eastin--Knill theorem}
\label{app:ek}

Here we derive specific lower bounds on the infidelity of codes covariant with respect to unitary groups which lead to new approximate Eastin--Knill theorems, following the discussion in \cite{faist2019continuous}. 

$SU(d_L)$-covariant codes in an $n$-partite system $S$ are defined by the encoding channels $\mE_{\LtoS}$ which satisfy
\begin{equation}
\mE_{\LtoS}\big(U_L(g)(\cdot)U_L^\dagger(g)\big) = 
\bigg( \bigotimes_{k=1}^n U_{S_k}(g) \bigg) \mE_{\LtoS}(\cdot) \bigg( \bigotimes_{k=1}^n U_{S_k}^\dagger(g)\bigg),\;\forall g\in SU(d_L),
\end{equation}
where $U_{S_k}(g)$ and $U_L(g)$ are unitary representations of $SU(d_L)$. It was shown in Theorem 18 in the Supplemental Material of  \cite{faist2019continuous} that fixing $H_L = {\rm diag}(1,0,\ldots,-1)$ and letting $H_{S_k}$ be the corresponding generator acting on the subsystem $k$, we have
\begin{equation}
d_k \geq \binom{d_L - 1 + \lceil \norm{H_{S_k}} \rceil}{d_L - 1},
\end{equation}
where $\lceil\norm{H_{S_k}}\rceil$ denotes the closest integer no smaller than $\norm{H_{S_k}}$. 
Using the inequality $\binom{a+b}{a} \geq (1+\frac{b}{a})^a$, 
\begin{gather}
d_k \geq \left(\frac{d_L - 1 + \lceil \norm{H_{S_k}} \rceil}{d_L - 1}\right)^{d_L - 1}, ~~\Rightarrow~~ \left(\exp\left(\frac{\ln d_k}{d_L - 1}\right) - 1\right)(d_L - 1) \geq \norm{H_{S_k}},\\
\Rightarrow~~
\sum_{k=1}^n \Big(\exp\left(\frac{\ln d_k}{d_L - 1}\right) - 1\Big)^2 (d_L - 1)^2\geq \frac{1}{4}\sum_k (\Delta H_{S_k})^2. 
\end{gather}
Then using \eqref{eq:eb-3}, we have for any $\epsilon \geq \epsilon(\mM_S,\mE_{\LtoS})$,
\begin{equation}
\epsilon 
\geq  \ell_1\left(\frac{1}{4n\sum_{k=1}^n \big(\exp\big(\frac{\ln d_k}{d_L - 1}\big) - 1\big)^2 (d_L - 1)^2}\right).
\end{equation}
For large $d_L = \Omega(\ln d_k)$, 
\begin{equation}
\begin{split}
\epsilon 
&\geq  \ell_1\left(\frac{1}{4n \sum_{k=1}^n (\ln d_k)^2} + O\left(\frac{1}{n d_L \sum_{k=1}^n \ln d_k}\right)\right) \\ &= \frac{1}{4n \sum_{k=1}^n (\ln d_k)^2} + O\left(\frac{1}{n d_L \sum_{k=1}^n \ln d_k} + \frac{1}{n^2 \left(\sum_{k=1}^n (\ln d_k)^2\right)^2} \right).
\end{split}
\end{equation}
Compared to Theorem 4 in \cite{faist2019continuous}: 
\begin{equation}
\epsilon \geq \left( \frac{1}{2n \max_k \ln d_k}+ O\left(\frac{1}{n d_L}\right) \right)^2 = \frac{1}{4n^2 (\max_k \ln d_k)^2} + O\left(\frac{1}{n^2 d_L \max_k \ln d_k}\right), 
\end{equation}
our bound improves the maximum of $\ln d_k$ in the denominator to their quadratic mean. 
Moreover, it works for not only single-error erasure noise channel $\mM_S = \sum_{k=1}^n \frac{1}{n} \mM_{S_k}$ where $\mM_{S_k}(\cdot) = \ket{\vac}\bra{\vac}_{S_k}$, but also single-error depolarizing noise channel $\mM_S = \sum_{k=1}^n \frac{1}{n} \mM_{S_k}$ where $\mM_{S_k}(\cdot) = \frac{\id}{d_k}$.

\section{Infidelity of thermodynamic codes under depolarizing noise}
\label{app:infidelity}
 
Here we use Corollary 2 from~\cite{beny2010general} to calculate the infidelity of thermodynamic codes under depolarizing noise channels in the limit $m/N \rightarrow 0$: 
\begin{lemma}[\cite{beny2010general}]
\label{lemma:infidelity}
A code defined by its projector $P$ is $\epsilon$-correctable under a noise channel $\mM(\cdot) = \sum_{i=1}^r K_i(\cdot)K_i^\dagger$ if and only if $P K_i^\dagger K_j P = A_{ij} P + P \delta A_{ij} P$ for some $A_{ij}$ and $\delta A_{ij}$ where $A_{ij}$ are the components of a density operator, 
and $1 - f^2(\mA+\delta\mA,\mA) \leq \epsilon$ where $\mA(\rho) = \sum_{ij}A_{ij}\trace(\rho)\ket{i}\bra{j}$ and $(\mA + \delta\mA)(\rho) = \mA(\rho) + \sum_{ij}\trace(\rho  \delta A_{ij})\ket{i}\bra{j}$.  
\end{lemma}

Let $P = \ket{g_0}\bra{g_0} + \ket{g_1}\bra{g_1}$, $\mM = \mM_S$ with Kraus operators 
\begin{equation}
K_{k,i} = \frac{1}{2\sqrt{n}} (U_{i})_{S_k},\quad i=0,1,2,3,
\end{equation}
where $U_0,U_1,U_2,U_3$ are respectively $\id,\sigma_x = \ket{1}\bra{-1}+\ket{-1}\bra{1},\sigma_y = -i\ket{1}\bra{-1}+i\ket{-1}\bra{1},$ and $\sigma_z = \ket{1}\bra{1} - \ket{-1}\bra{-1}$. 

For $m\geq 3$, $\braket{g_0|E|g_1} = 0$ for any operator $E$ acting on at most two qubits.  Here we consider $\delta A_{ij} \propto (\ket{g_0}\bra{g_0} - \ket{g_1}\bra{g_1})$.
That is, let $\delta A_{ij} = B_{ij} (\ket{g_0}\bra{g_0} - \ket{g_1}\bra{g_1})$. $A$ and $B$ are $4n \times 4n$ matrices 
\begin{equation}
A = 
\begin{pmatrix}
A^{(0,0)} & A^{(0,1)} & A^{(0,2)} & A^{(0,3)}\\
A^{(1,0)} & A^{(1,1)} & A^{(1,2)} & A^{(1,3)}\\
A^{(2,0)} & A^{(2,1)} & A^{(2,2)} & A^{(2,3)}\\
A^{(3,0)} & A^{(3,1)} & A^{(3,2)} & A^{(3,3)}\\
\end{pmatrix},\quad
B = 
\begin{pmatrix}
B^{(0,0)} & B^{(0,1)} & B^{(0,2)} & B^{(0,3)}\\
B^{(1,0)} & B^{(1,1)} & B^{(1,2)} & B^{(1,3)}\\
B^{(2,0)} & B^{(2,1)} & B^{(2,2)} & B^{(2,3)}\\
B^{(3,0)} & B^{(3,1)} & B^{(3,2)} & B^{(3,3)}\\
\end{pmatrix},
\end{equation}
where 
\begin{gather}
A^{(i,j)}_{kk'} = \frac{1}{2}(\braket{g_0|K_{k,i}^\dagger K_{k',j}|g_0} + \braket{g_1|K_{k,i}^\dagger K_{k',j}|g_1}),\\   
B^{(i,j)}_{kk'} = \frac{1}{2}(\braket{g_0|K_{k,i}^\dagger K_{k',j}|g_0} - \braket{g_1|K_{k,i}^\dagger K_{k',j}|g_1}), 
\end{gather}
so that $P K_i^\dagger K_j P = A_{ij} P + P \delta A_{ij} P$ holds.

\begingroup
\allowdisplaybreaks
A detailed calculation shows that $A^{(i,j)} = 0$ when $i \neq j$, $B^{(i,j)} = 0$ when $i+j \neq 3$, and 
\begin{gather}
\label{eq:AB-1}
A^{(0,0)} = \frac{1}{4n}
\begin{pmatrix}
1 & 1 & \cdots & 1\\
1 & 1 & \cdots & 1\\
\vdots & \vdots & \ddots & \vdots\\
1 & 1 & \cdots & 1\\
\end{pmatrix},
\\
A^{(1,1)} = A^{(2,2)} = \frac{1}{4n}
\begin{pmatrix}
1 & \frac{n^2-m^2}{2n(n-1)} & \cdots & \frac{n^2-m^2}{2n(n-1)}\\
\frac{n^2-m^2}{2n(n-1)} & 1 & \cdots & \frac{n^2-m^2}{2n(n-1)}\\
\vdots & \vdots & \ddots & \vdots\\
\frac{n^2-m^2}{2n(n-1)} & \frac{n^2-m^2}{2n(n-1)} & \cdots & 1\\
\end{pmatrix} ,
\\
\label{eq:AB-2}
A^{(3,3)} = \frac{1}{4n}
\begin{pmatrix}
1 & \frac{m^2-n}{n(n-1)} & \cdots & \frac{m^2-n}{n(n-1)}\\
\frac{m^2-n}{n(n-1)} & 1 & \cdots & \frac{m^2-n}{n(n-1)}\\
\vdots & \vdots & \ddots & \vdots\\
\frac{m^2-n}{n(n-1)} & \frac{m^2-n}{n(n-1)} & \cdots & 1\\
\end{pmatrix},
\\
\label{eq:AB-3}
B^{(0,3)} = B^{(3,0)} = 
\frac{m}{4n^2}
\begin{pmatrix}
1 & 1 & \cdots & 1\\
1 & 1 & \cdots & 1\\
\vdots & \vdots & \ddots & \vdots\\
1 & 1 & \cdots & 1\\
\end{pmatrix},\quad
B^{(1,2)} = -B^{(2,1)} = 
i \frac{m}{4n^2} \id. 
\end{gather}
Next we note that 
\begin{equation}
\begin{split}
f(\mA,\mA+\delta\mA) 
&= \min_{\ket{\psi}}
f((\mA\otimes\id_R)(\ket{\psi}\bra{\psi}), ((\mA +\delta\mA)\otimes\id_R)(\ket{\psi}\bra{\psi}))\\
&= \min_{p_i,\rho_i,i=0,1}
f( A \otimes (p_0\rho_0+p_1\rho_1), p_0 (A+B) \otimes \rho_0 + p_1 (A-B) \otimes \rho_1)\\
&\geq \min_{p_i,\rho_i,i=0,1} p_0 f(A,A+B) + p_1 f(A,A-B) = f(A,A+B),
\end{split}
\end{equation}
where in the second step we define $ \braket{g_i|\psi}\braket{\psi|g_i} = p_i\rho_i$ for $i = 0,1$, and in the third step we use the joint concavity of fidelity and in the last step we use $f(A+B) = f(A-B)$. Therefore we must have 
\begin{equation}
f(\mA,\mA + \delta \mA) = f(A,A+B),
\end{equation}
by noticing that $f(\mA(\ket{g_0}\bra{g_0}),(\mA + \delta \mA)(\ket{g_0}\bra{g_0})) = f(A,A+B)$. 
\endgroup
\begingroup
\allowdisplaybreaks
First note that $A^{(i,i)}$ and $B^{(i,j)}$ could be diagonalized in the following way: 
\begin{gather}
A^{(0,0)} = \frac{1}{4n} (n\ket{\psi_1}\bra{\psi_1}),\quad B^{(0,3)} = B^{(3,0)} = \frac{m}{4n} \ket{\psi_1}\bra{\psi_1}, \\
A^{(1,1)} = A^{(2,2)} = \frac{1}{4n} \bigg(\frac{n^2+2n-m^2}{2n}\ket{\psi_1}\bra{\psi_1} + \frac{n^2-2n+m^2}{2n(n-1)} \sum_{k=2}^n \ket{\psi_k}\bra{\psi_k}\bigg),\\
A^{(3,3)} = \frac{1}{4n}\bigg( \frac{m^2}{n}\ket{\psi_1}\bra{\psi_1} + \frac{n^2-m^2}{n(n-1)} \sum_{k=2}^n \ket{\psi_k}\bra{\psi_k}\bigg),
\end{gather}
\endgroup
where $\ket{\psi_1} = \frac{1}{\sqrt{n}} 
(1~1~\cdots~1)$ and $\{\ket{\psi_k}\}_{k>1}$ is an arbitrary orthonormal basis of the orthogonal subspace of $\ket{\psi_1}$.  
Since $A^{(i,j)} = A^{(j,i)} = B^{(i,j)} = B^{(j,i)}= 0$ when $i\in\{1,2\}$ and $j \in \{0,3\}$, we have 
\begin{equation}
f(A,A+B) = f(A^{(0)},A^{(0)}+B^{(0)}) + f(A^{(1)},A^{(1)}+B^{(1)}),
\end{equation}
where 
\begin{equation}
(\cdot)^{(0)} = 
\begin{pmatrix}
(\cdot)^{(0,0)} & (\cdot)^{(0,3)}\\
(\cdot)^{(3,0)} & (\cdot)^{(3,3)}\\
\end{pmatrix},\quad 
(\cdot)^{(1)} = 
\begin{pmatrix}
(\cdot)^{(1,1)} & (\cdot)^{(1,2)}\\
(\cdot)^{(2,1)} & (\cdot)^{(2,2)}\\
\end{pmatrix}.
\end{equation}
We first calculate $f(A^{(0)},A^{(0)}+B^{(0)})$. We have 
\begin{multline}
(A^{(0)})^{1/2}(A^{(0)}+B^{(0)})(A^{(0)})^{1/2}
= \\ \begin{pmatrix}
\frac{1}{4} \\
\frac{m^2}{4n^2} \\
\end{pmatrix} 
\begin{pmatrix}
\frac{1}{4} & 
\frac{m^2}{4n^2} \\
\end{pmatrix} \otimes \ket{\psi_1}\bra{\psi_1}
+ \begin{pmatrix}
0 & 0 \\
0 & \big(\frac{n^2-m^2}{4n^2(n-1)}\big)^2\\
\end{pmatrix}
\otimes \sum_{k=2}^n \ket{\psi_k}\bra{\psi_k}. 
\end{multline}
Then 
\begin{equation}
\begin{split}
f(A^{(0)},A^{(0)}+B^{(0)}) 
&= \trace\Big(\big((A^{(0)})^{1/2}(A^{(0)}+B^{(0)})(A^{(0)})^{1/2}\big)^{1/2}\Big) \\
&= \sqrt{\frac{1}{4^2} + \Big(\frac{m^2}{4n^2}\Big)^2} + \frac{n^2-m^2}{4n^2} = \frac{1}{2} - \frac{m^2}{4n^2} + O\Big(\frac{m^4}{n^4}\Big). 
\end{split}
\end{equation}
In order to calculate $f(A^{(0)},A^{(0)}+B^{(0)})$, 
we first note that 
\begin{equation}
\begin{split}
(A^{(1)})^{1/2}(A^{(1)}+B^{(1)})(A^{(1)})^{1/2}
= \begin{pmatrix}
(A^{(1,1)})^2 & 0 \\
0 & (A^{(1,1)})^2 \\ 
\end{pmatrix} + 
\begin{pmatrix}
0 & i\frac{m}{4n^2} A^{(1,1)} \\
-i\frac{m}{4n^2} A^{(1,1)} & 0 \\ 
\end{pmatrix}.
\end{split}
\end{equation}
Then we use the Taylor expansion formula for square root of positive matrices:  $\sqrt{\Lambda^2+Y} = \Lambda + \chi[Y] - \chi(\chi[Y]^2) + O(Y^3)$ for any positive diagonal matrix $\Lambda$ and small $Y$~\cite{del2018taylor}, where 
\begin{equation}
\chi[(\cdot)]_{ij} = \frac{(\cdot)_{ij}}{\Lambda_i + \Lambda_j}.
\end{equation}
Let $A^{(1)} = \Lambda$ such that $\Lambda_1 = \frac{n^2+2n-m^2}{8n^2}$ and $\Lambda_{k} = \frac{n^2-2n+m^2}{8n^2(n-1)}$ for $k > 1$, we find that 
\begin{equation}
f(A^{(1)},A^{(1)}+B^{(1)}) = 
\frac{1}{2} - \Big(\frac{m}{4n^2}\Big)^2\sum_{k=1}^n \frac{1}{4\Lambda_k} + O\Big(\frac{m^3}{n^3}\Big) = \frac{1}{2} - \frac{m^2}{8n^2} + O\Big(\frac{m^3}{n^3}\Big). 
\end{equation}
Therefore 
\begin{equation}
1 - f(\mA,\mA+\delta\mA)^2 = 1 - f(A,A+B)^2 = \frac{3m^2}{4n^2} + O\Big(\frac{m^3}{n^3}\Big), 
\end{equation}
which serves as an upper bound on the infidelity of thermodynamic codes under depolarizing noise due to \lemmaref{lemma:infidelity}.

\bibliographystyle{aps}
\bibliography{refs-covariant-final}

\end{document}